\documentclass[journal]{vgtc}

\onlineid{1613}

\vgtccategory{Research}

\title{EXaCTz: Guaranteed Extremum Graph and Contour Tree Preservation for Distributed- and GPU-Parallel Lossy Compression}

\author{
  \authororcid{Yuxiao Li}{0000-0002-8715-5982},
  \authororcid{Mingze Xia}{0009-0000-7653-9000},
  \authororcid{Xin Liang}{0000-0002-0630-1600},
  \authororcid{Bei Wang}{0000-0002-9240-0700},
  and \authororcid{Hanqi Guo}{0000-0001-7776-1834}
}

\authorfooter{
  \item Yuxiao Li and Hanqi Guo are with The Ohio State University. 
  E-mail:\{li.14025, guo.2154\}@osu.edu.
 \item Mingze Xia and Xin Liang are with Oregon State University.
 E-mail: \{xiami, lianxin\}@oregonstate.edu.
 \item Bei Wang is with the University of Utah.
 E-mail: beiwang@sci.utah.edu.
}

\abstract{This paper introduces EXaCTz, \updated{a parallel algorithm that corrects lossy-compressed scalar field data to preserve extremum graphs and contour trees concurrently.} While error-bounded lossy compression is essential for large-scale scientific simulations and workflows, existing topology-preserving methods suffer from (1) a significant throughput disparity, where topology correction speeds are on the order of MB/s, lagging orders of magnitude behind compression speeds on the order of GB/s; (2) limited support for diverse topological descriptors; and (3) a lack of theoretical convergence bounds. To address these challenges, EXaCTz introduces a high-performance, bounded-iteration algorithm that enforces topological consistency by deriving targeted edits for decompressed data. Unlike prior methods that rely on explicit topology reconstruction, EXaCTz enforces consistent min/max neighbors of all vertices, along with global ordering among critical points. As such, the algorithm enforces consistent critical-point classification, saddle--extremum connectivity, and the preservation of merge/split events. We theoretically prove the convergence of our algorithm, bounded by the longest path in a \textit{vulnerability graph} that characterizes potential cascading effects during correction. Experiments on real-world datasets show that EXaCTz achieves a single-GPU throughput of up to 4.52 GB/s, outperforming the state-of-the-art contour-tree-preserving method (Gorski et al.~\cite{gorski2025general}) by up to 213$\times$ (with a single-core CPU implementation for fair comparison) and 3,285$\times$ (with a single-GPU version). In distributed environments, EXaCTz scales to 128 GPUs with 55.6\% efficiency (compared with 6.4\% for a naive parallelization), processing datasets of up to 512 GB in under 48 seconds and achieving an aggregate correction throughput of up to 32.69 GB/s.
}


\nocopyrightspace

\graphicspath{{figures/}} 
\usepackage{amsmath, amsthm}

\theoremstyle{plain} 
\newtheorem{theorem}{Theorem}[section]

\usepackage{algorithm}
\usepackage{algpseudocode}
\theoremstyle{definition} 

\theoremstyle{remark} 

\usepackage{booktabs}
\usepackage{graphicx}
\usepackage{amssymb}
\usepackage{multirow}
\usepackage{tikz}
\usepackage{mdframed}
\usepackage{wrapfig}
\newcommand{\tool}{EXaCTz}
\newcommand{\updated}[1]{\textcolor{black}{#1}}

\usetikzlibrary{shapes.geometric, arrows.meta, positioning, calc, decorations.pathmorphing, decorations.markings}

\usepackage{mathptmx}                  

\begin{document}
\maketitle


\section{Introduction}

With the rapid growth of data from large-scale scientific simulations (e.g., reaching the order of terabytes in cosmology, combustion, and climate modeling~\cite{Cappello19,SDRBench}), researchers face increasing challenges in storage, transmission, and downstream analysis~\cite{di2024survey}. To reduce data volumes, error-bounded lossy compressors (e.g., SZ~\cite{sz, sz3, Tao_2017, Kai2021, Kai2020} and ZFP~\cite{zfp}) have become widely adopted due to the ability to achieve high compression ratios while providing strict pointwise error guarantees. To keep up with increasing data generation rates, recent studies leverage GPU and distributed parallelism to improve compression throughput for performance-critical scientific applications. 
For example, cuSZp~\cite{huang2023cuszp, huang2024cuszp2} achieves tens of GB/s throughput, enabling in situ compression at scale.

More recently, topology-preserving compression methods have been developed to address distortions in topological features in lossy-compressed data~\cite{li_msz, toposz, gorski2025general, li2026pmsz, Xia25, Xia24, li2025, liange_2020, LiangDCRLOCPG23}, because even with strong, pointwise quality guarantees, error-bounded lossy compressors do not guarantee the correctness of features required for downstream analyses in applications such as chemistry, combustion, and cosmology~\cite{Bhatia_2018, Gunther2014, Dora_2013, Chen2009, Shivashankar_2016}. For example, extremum graphs~\cite{extremumgraph} and contour trees~\cite{CARR200375} are important topological descriptors of scalar fields. An extremum graph captures the connectivity of critical points, while a contour tree represents the hierarchical evolution of level sets (isocontours) across scalar values~\cite{CARR200375}. In cosmology simulations, extremum graphs extract the spatial connectivity of dark matter halos, while contour trees capture the hierarchical structure of the cosmos~\cite{Shivashankar_2016}. Distortions caused by compression errors, as shown in Figure~\ref{fig:distortions}, can lead scientists to falsely identify dark matter structures or miscalculate the accretion history, affecting downstream scientific conclusions. 

\begin{figure}[htb!]
    \centering
    \includegraphics[width=\linewidth]{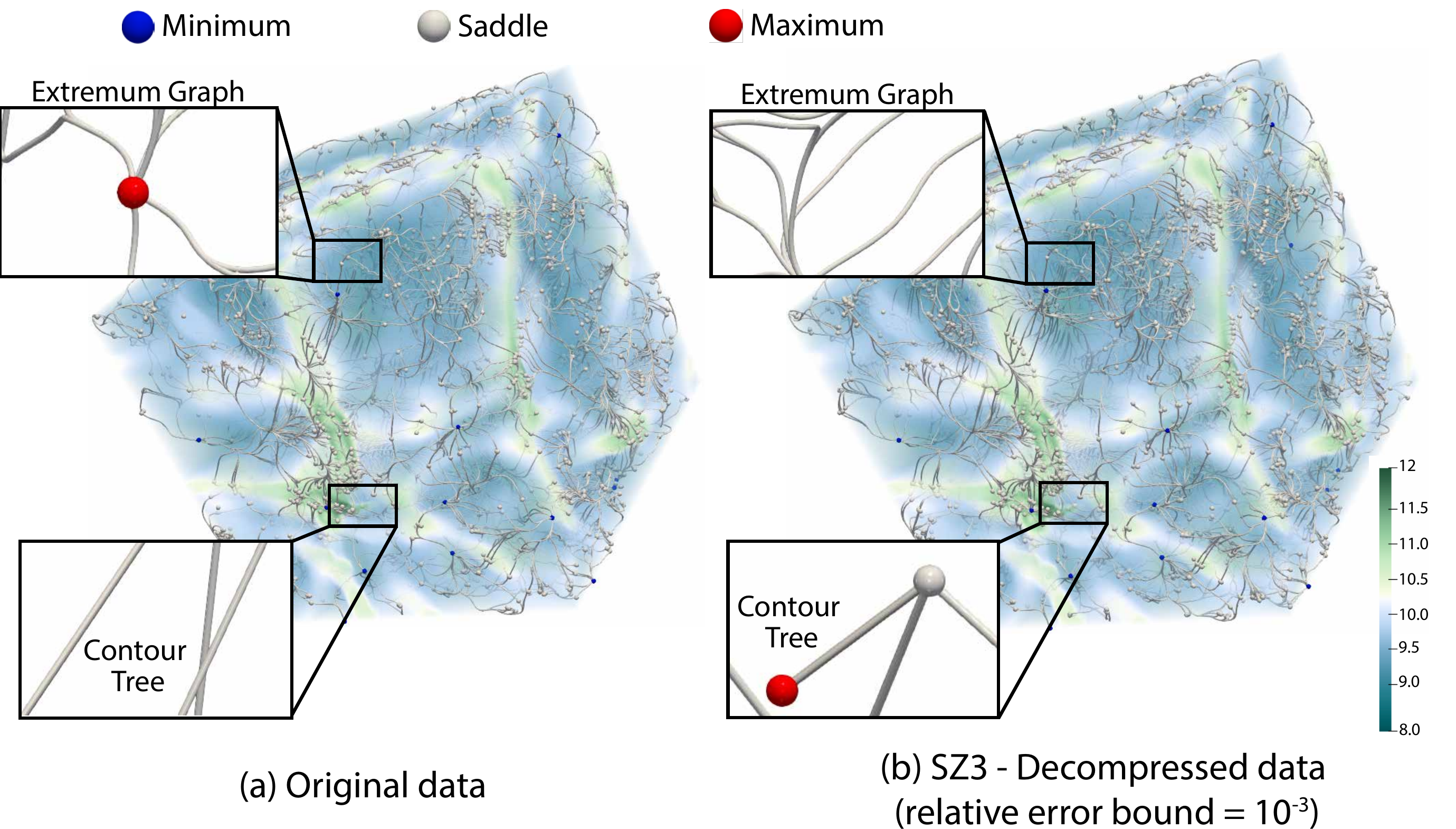}
    \caption{Distortions introduced by error-bounded lossy compression on the cosmology dataset. (a) Original data. (b) SZ3-decompressed data. }
    \label{fig:distortions}
    \vspace{-0.1in}
\end{figure}

While effective in preserving topology, existing topology-preserving compression methods have three major limitations: (1) a significant performance gap compared to modern high-throughput compressors, with topology correction throughput on the order of MB/s versus GB/s for state-of-the-art compressors;  (2) limited support for diverse topological descriptors, leaving important structures such as extremum graphs unaddressed; and (3) no explicit upper bound on the number of correction iterations, making worst-case performance unpredictable.

First, the performance gap remains a critical bottleneck for topology-preserving compression.
While modern compressors achieve throughput at the GB/s scale, existing methods typically operate at MB/s. This disparity is especially severe for contour-tree-preserving methods, which rely on explicit contour tree construction. For example, recent work by Gorski et al.~\cite{gorski2025general} (referred to as TopoA) improves the efficiency of contour-tree-preserving compression by avoiding full contour tree reconstruction in TopoSZ~\cite{toposz}, but still requires repeated partial tree computations. As a result, preserving the simplified contour tree of a $512^3$ cosmology snapshot takes over $16$ minutes, yielding throughput at the MB/s scale, far below modern compressors that achieve tens of GB/s. Contour tree computation and topological simplification pose significant performance and scalability challenges~\cite{Carr2022Distributed, Mingzhe25} due to their reliance on global data dependencies. As demonstrated by our experiments, the simplification phase introduces substantial overhead (e.g., 3.21 s for SZ3 compression vs. 72.52 s for simplification on a $512^3$ cosmology snapshot). Methods that preserve simplified topology and rely on explicit topology construction inherit these scalability limitations.

Second, existing methods provide limited support for preserving extremum graphs. 
Most prior work focuses on a single topological descriptor, such as contour trees~\cite{toposz, gorski2025general} or Morse-Smale segmentations~\cite{li_msz, li2026pmsz}, without explicitly preserving extremum graphs. However, extremum graphs encode the connectivity among critical points, which is essential for capturing structural relationships in scalar fields. Distortions in extremum graphs could lead to incomplete or inconsistent topological representations, as illustrated earlier.

Third, existing iterative topology correction methods lack explicit theoretical guarantees on the number of iterations. While prior approaches~\cite{li_msz, li2026pmsz, li2025} establish convergence under a finite number of iterations, they do not provide explicit bounds on the required number of iterations, making their worst-case performance unpredictable.

\updated{To address these challenges, we introduce EXaCTz, a distributed- and GPU-parallel algorithm that builds on the edit-based correction paradigm}~\cite{li_msz, li2026pmsz, li2025} to concurrently preserve extremum graphs and contour trees under error-bounded lossy compression, without explicit topology construction. 
Our approach addresses the three limitations through the following key ideas.

First, to bridge the performance gap, we avoid explicit topology construction and instead enforce \updated{a set of sufficient (but not necessary)} topology-preserving conditions on scalar values, enabling efficient GPU acceleration and distributed execution. 

Second, to enable concurrent preservation of extremum graphs and contour trees, we use the theoretical connection between them established by ExTreeM~\cite{ExtreeM}, which derives contour trees from extremum graphs. 
Specifically, we enforce (1) extremum graph preservation (EG constraints), (2) scalar ordering among saddles preservation (saddle constraints), and (3) merge and split events associated with saddles (event constraints). 
We use an \textit{edit-based iterative correction} strategy~\cite{li_msz} to enforce these constraints by iteratively detecting violations and applying bounded, monotonic edits to \updated{a subset of data values within the prescribed error bound with respect to the original data.}

Third, to provide theoretical guarantees, we analyze the worst-case behavior of the iterative correction process by modeling cascading effects. We characterize the cascading effects using a \textit{vulnerability graph} \updated{that captures dependencies among corrections, which is formally defined in Section}~\ref{sec:convergence}. Based on this formulation, we derive an explicit upper bound on the number of correction iterations via the longest path in the graph, ensuring predictable convergence and computational cost.

To efficiently handle datasets that exceed single-node memory capacity or are distributed across multiple nodes, we reformulate event constraints to improve scalability in distributed-memory environments. In the serial version formulation, event constraints require identifying the extrema connected to each saddle through integral paths, i.e., tracing steepest ascent and descent trajectories in the scalar field. Global integral path tracing is difficult to scale in distributed-memory environments due to cross-partition dependencies~\cite{li_msz, Will2024, li2026pmsz}. To eliminate this communication bottleneck, we replace explicit integral path tracing with strict enforcement of the global scalar ordering among all critical points, sufficient to ensure correct topological events without requiring path computation. Our contributions are summarized as follows:
\begin{itemize}
    \item \updated{We derive \updated{a set of sufficient (but not necessary)} conditions that guarantee concurrent preservation of \textbf{extremum graphs} and \textbf{contour trees} under error-bounded lossy compression, without explicit topology construction.}
    \item We introduce EXaCTz, a \textbf{distributed-} and \textbf{GPU-parallel} algorithm that preserves both extremum graphs and contour trees under error-bounded lossy compression, achieving GB/s-scale throughput on GPUs and substantially reducing runtime compared to state-of-the-art methods (e.g., up to $213\times$ speedup compared to prior CPU-based methods in serial experiments).
    \item We formally establish the convergence of the edit-based correction strategy and derive a worst-case \textbf{upper bound} on the number of correction iterations.
    \item We improve the scalability of our method in distributed-memory environments by avoiding computationally expensive integral path tracing, increasing parallel efficiency from $6.4\%$ to $55.6\%$, and enabling GB/s-scale throughput (up to 32.69 GB/s) while processing $512$~GB datasets in under 48 seconds.    
\end{itemize}

\section{Related Work}

We review related work on error-bounded lossy compression for scientific data, topology-preserving lossy compression and merge (and contour) tree computation.

\subsection{Error-Bounded Lossy Compression for Scientific Data}
Error-bounded lossy compression has been widely used in large-scale scientific computing as it significantly reduces data size while guaranteeing that the pointwise error does not exceed a user-defined bound. However, most existing compressors only focus on preserving pointwise numerical accuracy, without ensuring the preservation of features of interest that are critical for scientific analysis. For a comprehensive survey, we refer readers to Di et al.~\cite{di2024survey}.

Error-bounded lossy compressors can be categorized based on the optimization objectives in terms of compression ratio and throughput. Compressors optimized for high compression ratio aim to improve data reduction efficiency. 
Representative methods include the SZ series~\cite{sz, sz3, Tao_2017, Kai2021}, its variants~\cite{liu2022qoz,jiaoipdps25,jiaoipdps26,jiaovldb22}, and neural-network-enhanced approaches such as AE-SZ~\cite{liu_2021}, SRNN-SZ~\cite{liu2023srnsz}, and CFNN~\cite{liusc24w,liuhpdc25}, which use advanced predictors to estimate data values more accurately. 
GPU-based extensions such as cuSZ-i and cuSZ-Hi~\cite{liu2024cuszi, liu2025cuszhi} further improve compression ratio through multi-level interpolation and integrated lossless encoding. 
Compressors such as ZFP~\cite{zfp}, TTHRESH~\cite{TTHRESH}, SPERR~\cite{SPERR}, and MGARD~\cite{MGARD} also achieve high compression ratios for scientific data by leveraging transform-based representations.
Compressors optimized for high throughput focus on achieving fast end-to-end performance. 
For example, cuSZp~\cite{huang2023cuszp, huang2024cuszp2} achieves high-throughput compression by fusing multiple kernels and minimizing memory overhead, targeting performance-critical scientific workflows.
 
\subsection{Topology-Preserving Lossy Compression}
Topology-preserving compression has been explored for a variety of topological structures in scalar-field data, such as Morse-Smale complexes~\cite{li2025} (and their induced segmentations~\cite{li_msz}) and merge/contour trees~\cite{topoqz, gorski2025general, toposz}. However, existing methods are designed to preserve a single topological descriptor rather than multiple simultaneously, and simply storing topological structures as metadata is insufficient, since downstream analyses also rely on the underlying scalar field values, making consistency between data and topology essential.

\noindent\textbf{Contour Tree Preservation} has been studied in the context of error-bounded lossy compression, where most existing methods rely on explicit contour tree construction during compression, introducing significant overhead and limiting scalability on large datasets.
Soler et al.~\cite{topoqz} provided one of the earliest topology-aware compression methods, preserving critical point pairs above a persistence threshold~$\epsilon$ through simplification and quantization, but not designed to preserve contour trees.
Lin et al. proposed TopoSZ~\cite{toposz} by modifying SZ1.4~\cite{sz} to enforce upper and lower bounds derived from contour-tree-induced segmentation, ensuring that the decompressed data retains the same contour tree as the original. 
Gorski et al.~\cite{gorski2025general} further generalized localized error bounds by augmenting a wide range of compressors (including SZ, ZFP, and TTHRESH), and introduce a progressive bound tightening strategy that constructs partial contour trees in each iteration.

\noindent\textbf{Morse--Smale Segmentations/Complexes Preservation} has also been studied in topology-preserving compression, but does not directly address the preservation of contour trees. Li et al.~\cite{li_msz, li2026pmsz} proposed an edit-based strategy to preserve Morse--Smale segmentations by modifying the scalar values at a subset of data points in the decompressed data. The edit-based strategy was later extended to preserve the full Morse--Smale complexes~\cite{li2025} by designing a different editing scheme. 

Related work has also explored topology-preserving compression for vector fields. 
Early work focused on modifying data or meshes to maintain topological consistency, such as collapsing mesh edges in 2D vector fields~\cite{Tricoche_01} and simplifying vector field topology prior to compression~\cite{Theisel_03}. 
Recent approaches integrate topological constraints directly into the compression process. For example, Liang et al.~\cite{liange_2020, LiangDCRLOCPG23} and Xia et al.~\cite{Xia24, xia26,Xia25} proposed compressor-specific methods preserving critical points, their temporal trajectories, and separatrices.

\subsection{Merge and Contour Tree Computation}
Merge and contour trees capture the connectivity of (sub)level sets in scalar fields and are widely used in topological data analysis. However, the computation of merge trees is known to be expensive, especially for large-scale scientific data, becoming a potential bottleneck when incorporated into compression pipelines. We review merge tree computation algorithms, as contour tree construction relies on merge trees.

Most existing algorithms construct merge trees by processing the full scalar field and can be implemented as \textbf{serial methods}, such as union-find-based approaches~\cite{CARR200375, Kruskal56, CHIANG2005165} and integral-line-based methods~\cite{Maadasamy12, Ande_2023} which rely on global sorting or long integral paths tracing. Such global operations are limited by irregular memory accesses and severe data dependencies, making serial execution computationally expensive for large scale datasets. To improve computational efficiency, several \textbf{parallel methods} have been proposed. Shared-memory sweep-based algorithms, such as FTM-Tree~\cite{Gueunet17}, propagate superlevel-set components using priority queues. Data-parallel pruning approaches, including Parallel Peak Pruning (PPP)~\cite{Carr21}, leverage data-parallel primitives to iteratively remove branches. However, irregular data accesses, global updates, and dynamic connectivity changes limit the scalability of such techniques on GPUs. For datasets beyond the capacity of a single node, distributed approaches partition data across compute nodes. Carr et al.~\cite{Carr2022Distributed} introduced a distributed hierarchical representation, while Li et al.~\cite{Mingzhe25,Li2025Extremely} used distributed hypersweeps and branch decompositions. However, localized tree computations and subsequent global boundary merging still introduce heavy communication overhead. 

Alternatively, \textbf{extremum-graph-based} methods~\cite{ExtreeM} construct merge trees directly from extremum graphs and have been proven to produce the same merge trees as those computed from the full scalar field. We review it in the next section.

\section{Background}

We review extremum graphs and merge trees, the computation of merge trees via extremum graphs~\cite{ExtreeM}, and the widely used edit-based strategy for topology correction~\cite{li_msz}. The underlying data is modeled as a piecewise-linear (PL) scalar field $f: \mathbb{X} \rightarrow \mathbb{R}$ defined on a simply connected domain $\mathbb{X}$, where scalar values are assigned to the vertices.

\subsection{Extremum Graphs and Merge Trees}
\label{subsec:contour_trees}

We formally review integral paths, extremum graphs, and merge trees.

\noindent\textbf{Integral Paths} in a PL scalar field \updated{on regular grids} are discretized gradient flows represented by monotonic edge paths on the mesh, computed by iteratively tracing from a vertex to its highest (or lowest) adjacent neighbor until reaching a local maximum (or minimum). Integral paths are fundamental in computing extremum graphs, where descending and ascending paths are illustrated by the blue and red arrows in Figure~\ref{fig:EG}.

\noindent\textbf{Extremum Graphs}~\cite{extremumgraph} consist of nodes representing critical points (extrema and saddles) and edges representing integral paths (i.e., steepest ascent or descent paths) connecting critical points. Extremum graphs capture connectivity among extrema in the PL scalar field $f$. 
As shown in Figure~\ref{fig:EG}(a), the extremum graph for minima traces integral paths connecting local minima (blue nodes) and saddles (white nodes). The extremum graph for maxima (Figure~\ref{fig:EG}(b)) captures paths between local maxima (red nodes) and saddles.

\begin{figure}[htb!]
    \centering
    \includegraphics[width=\linewidth]{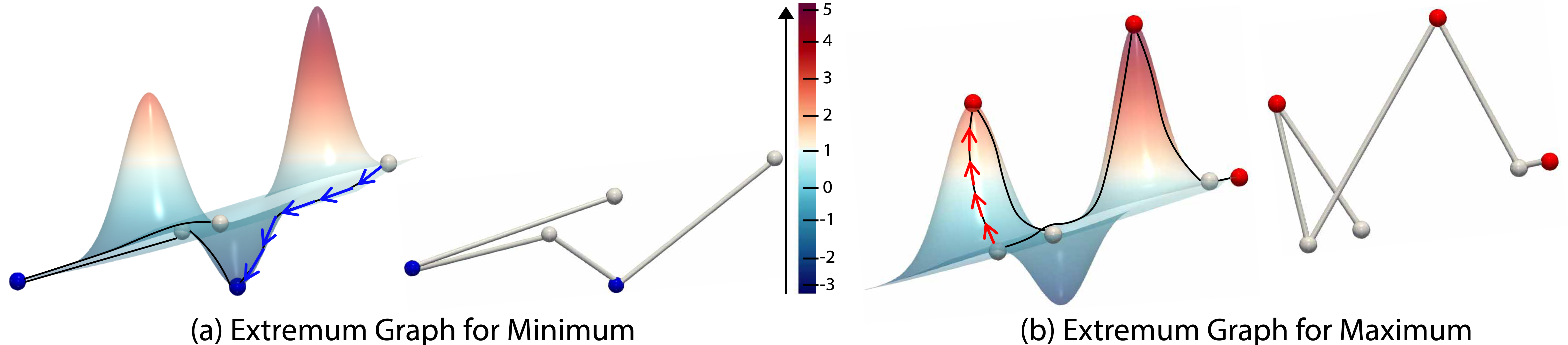}
    \vspace{-5mm}
    \caption{Illustration of integral paths and extremum graphs. (a) Extremum graph for minima captures the connectivity between minima (blue nodes) and saddles (white nodes). (b) Extremum graph for maxima captures the connectivity between maxima (red nodes) and saddles.}
    \label{fig:EG}
\end{figure}

\begin{figure}[htb!]
    \centering
    \includegraphics[width=\linewidth]{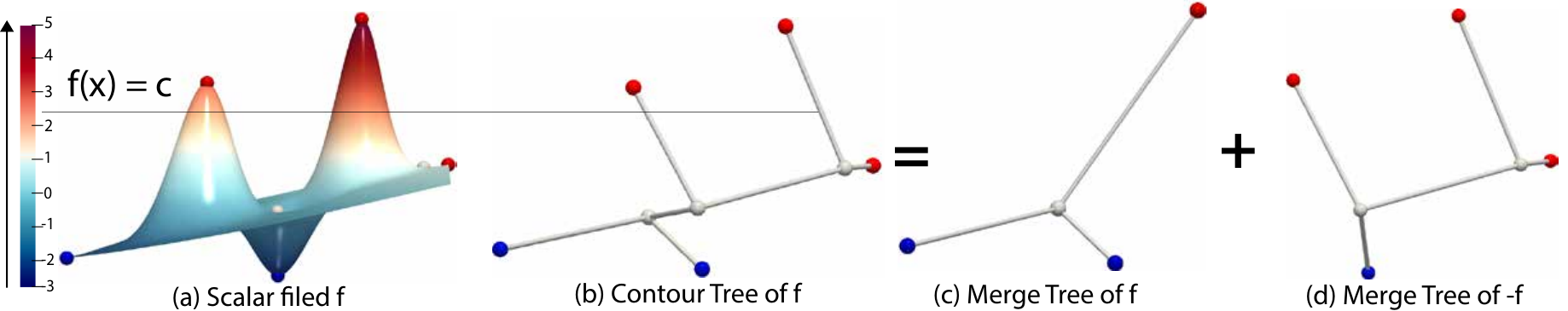}
    \caption{Illustration on merge/contour trees. 
    (a) Scalar field $f$ with a level set $L(c)$. (b) Contour tree. (c) Merge tree of $f$. (d) Merge tree of $-f$. }
    \label{fig:mt}
    \vspace{-0.1in}
\end{figure}

\noindent\textbf{Merge Trees} characterize the evolution of data connectivity as a scalar threshold $c$ varies over the data range, based on sublevel sets $L^{-}(c) = \{x \mid f(x) \le c\}$ and superlevel sets $L^{+}(c) = \{x \mid f(x) \ge c\}$. 
Under increasing thresholds, connected components in the sublevel sets are created, expanded, and merged at critical points (join saddles) in Figure~\ref{fig:mt}(c). 
Under decreasing thresholds, analogous behavior occurs for superlevel sets, as shown in the merge tree of $-f$ in Figure~\ref{fig:mt}(d).

\noindent\textbf{Contour Trees} track the connectivity of \textit{level sets} (or isocontours), defined as $L(c) = \{x \mid f(x) = c\}$ (illustrated as a horizontal slice in Figure~\ref{fig:mt}(a)), and provide an abstraction of how these level sets merge and split. In a contour tree (Figure~\ref{fig:mt}(b)), the leaf nodes correspond to extrema, while the non-leaf nodes correspond to saddles. Most contour tree algorithms first compute the merge trees of $f$ and $-f$ (also referred to as the join tree and split tree, respectively, in the literature), followed by a combination step to form the contour tree~\cite{CARR200375}.

\subsection{Merge Tree Computation via Extremum Graphs}
The theory of our contour-tree-correction algorithm is grounded by ExTreeM~\cite{ExtreeM}, an extremum-graph-based method that computes a merge tree equivalent to the one derived from the full scalar field.
The key idea of ExTreeM is to first extract all critical points in the scalar field and construct \updated{an} extremum graph,
and then to extract valid extremum-saddle pairs from the extremum graph to form the merge tree.
For brevity, we describe the construction of the join tree,  as the procedure for the split tree follows a completely symmetric rule. 

\noindent\textbf{Step 1. Construct Extremum Graphs}. 
ExTreeM first identifies critical points using local neighborhood comparisons. 
A vertex is a minimum (maximum) if the value is smaller (larger) than all neighboring vertices, while saddle points correspond to vertices whose neighborhoods partition into multiple monotonic regions (Figure~\ref{fig:saddle_class}), further categorized as join or split saddles based on whether sublevel-set 
components merge or superlevel-set components split.
ExTreeM then constructs \updated{an} extremum graph that contains extrema and saddles. For a join tree, the extremum graph consists of minima, join saddles, and edges connecting them (Figure~\ref{fig:EG}(a)). An edge is added from a join saddle $i$ to a minimum $j$ if there exists a neighboring vertex $k$ of $i$ such that $f_k < f_i$, and the steepest descending path from $k$ terminates at $j$. 
\begin{figure}[htb]
    \centering
    \vspace{-2mm}
    \includegraphics[width=\linewidth]{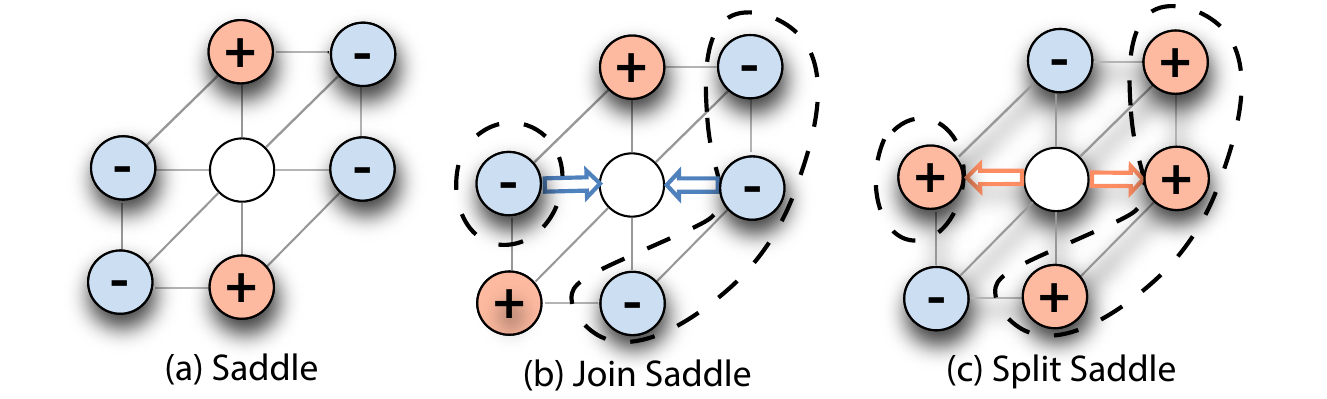}
    \vspace{-7mm}
     \caption{Saddle classifications based on local connectivity: (a) standard saddle, (b) join saddle, and (c) split saddle. $+$/$-$ indicate neighbors with higher/lower scalar values than the center point.}
    \label{fig:saddle_class}
    \vspace{-2mm}
\end{figure}

\begin{figure}[htb]
    \centering
    \includegraphics[width=\linewidth]{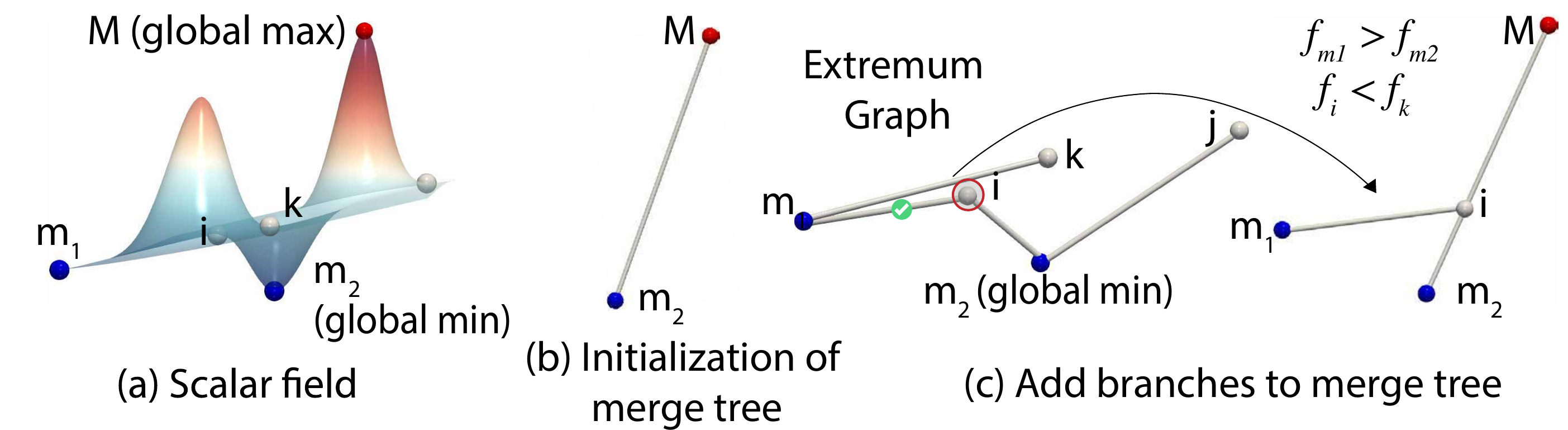}
    \vspace{-5mm}
    \caption{Merge tree construction from an extremum graph: (a) a scalar field $f$ with critical points (blue minima, white saddles); (b) initialization of the merge tree; (c) branch extraction via Extremum Graph Pairing (EGP).}
    \label{fig:jointree}
    \vspace{-0in}
\end{figure}

\noindent \textbf{Step 2. Compute Merge Trees from Extremum Graphs.} 
After extremum graph construction, ExTreeM applies \textit{Extremum Graph Pairing (EGP)} to iteratively extract tree branches (minimum--join-saddle pairs). 
As shown in Figure~\ref{fig:jointree}, the merge tree is initialized by connecting the global minimum ($m_2$) to the global maximum (Figure~\ref{fig:jointree}(b)). 
For an unpaired local minimum (e.g., $m_1$ in Figure~\ref{fig:jointree}(c)), EGP identifies the lowest connected saddle on the extremum graph (saddle $i$, since $f_i < f_k$). 
A branch is extracted when the minimum has the highest scalar value among all minima connected to that saddle. 
In Figure~\ref{fig:jointree}, saddle $i$ connects $m_1$ and $m_2$. Since $f_{m_1} > f_{m_2}$, the branch connecting $m_1$ and $i$ is extracted and added to the merge tree. After pairing, $m_1$ is removed and replaced by the lower minimum ($m_2$). The connectivity of the remaining saddles \updated{($k, j$)} is then updated. Once these saddles connect to a single minimum, distinct components no longer separate, so the saddles are removed. This process repeats until all minima are paired except the global minimum.

\noindent\textbf{Merge Tree Equivalence~\cite{ExtreeM}.} 
Let $G$ be \updated{an} extremum graph of scalar field $f$. ExTreeM establishes a strict bijection between the connected components of the sublevel set $L^-(c) = \{x \mid f(x) \le c\}$ and the components of the subgraph of $G$ induced by vertices with $f(v) \le c$, for any threshold $c$. Two minima belong to the same sublevel-set component if and only if they belong to the same component in $G$. Since merge trees are uniquely determined by the sequence of component creations (at minima) and merges (at saddles), this bijection guarantees that the merge tree extracted from $G$ is identical to that of the full field $f$. Symmetrical guarantees apply to the split tree; see~\cite{ExtreeM} for details.

\subsection{(Monotonic) Edit Strategy for Topology Correction}
\label{sec:bg_monotonic_edit}
To preserve topology in lossy compression, a series of recent studies~\cite{li_msz, li2026pmsz, li2025} focus on using a \textit{monotonic edit strategy} because topological distortions primarily stem from the disruption of local scalar orderings (e.g., an original ordering $f_v > f_u$ flipping to $\hat{f}_v < \hat{f}_u$\footnote{For adjacent vertices with equal scalar values (e.g., on plateaus), we use Simulation of Simplicity (SoS)~\cite{Edels94} to robustly break ties, treating the vertex with the larger global index as having the larger value.}). The core idea is to iteratively apply monotonic edits to vertices that cause inconsistency in topology. We use the discretized monotonic edit approach from DMTz~\cite{li2025} (a method designed to preserve Morse-Smale complexes). Specifically, the global absolute error bound $\xi$ is uniformly divided into $N$ discrete steps. In each iteration, the scalar value of the vertex causing a topological distortion is decreased by a fixed step size of $\Delta = \xi / N$. If a vertex $i$ requires more than $N$ edits, we store the edit in a lossless manner as the absolute lower bound, $f_i - \xi$. All edits are losslessly compressed (e.g., using ZSTD~\cite{zstd}). The choice of $N$ also balances computation and storage: smaller $N$ reduces iterations but increases lossless edits, while larger $N$ reduces lossless edits but increases iterations. In this paper, we set $N=5$ to balance computational and storage costs, based on our empirical evaluation on N (see Appendix~C). \updated{We assume $\xi \geq N \cdot \epsilon_{\text{mach}}$, where 
$\epsilon_{\text{mach}}$ denotes machine epsilon. Below this threshold, we skip correction and the compressor operates in lossless mode, as the error bound is too small for meaningful lossy compression. Consequently, no topological distortion is introduced, although this case rarely occurs in practice.}

\begin{figure}
    \centering
    \includegraphics[width=\linewidth]{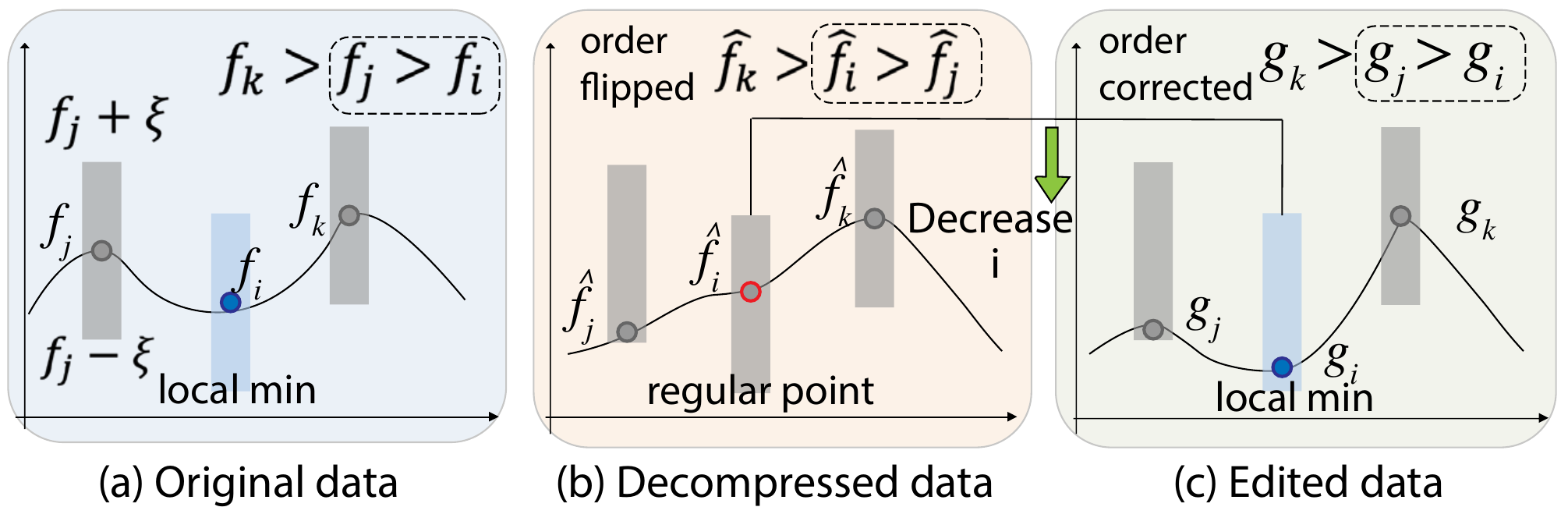}
   \caption{Monotonic edit strategy for correcting local ordering. (a) Original field with correct ordering. (b) Decompressed field with flipped ordering. (c) Edited field with restored ordering.}
    \label{fig:edit_strategy}
    \vspace{-0.1in}
\end{figure}

We take the preservation of local minima as an example. Suppose a true minimum $i$ ($f_k > f_j > f_i$, Figure~\ref{fig:edit_strategy}(a)) is inverted by compression errors such that $\hat{f}_k > \hat{f}_i > \hat{f}_j$ in Figure~\ref{fig:edit_strategy}(b). To restore the correct local ordering, we apply a decreasing edit to $i$ as shown in Figure~\ref{fig:edit_strategy}(c), producing a corrected value $g_i < g_j < g_k$ that still satisfies the user-defined error bound $\xi$ (i.e., $g_i > f_i - \xi$).
Note that modifying one vertex may introduce new violations in its one-hop neighborhood. The correction process iterates until all false cases are resolved.
\updated{To guarantee convergence, all edits are applied in a single direction (decreasing only) and bounded by $f_i - \xi$. If all values are reduced to their lower bounds, the original ordering is trivially preserved; thus, iteratively reducing values must eventually converge.} Building on this intuition, Section~\ref{sec:convergence} formally derives a theoretical iteration upper bound.

\section{Methodology}

This section presents our iterative workflow to preserve extremum graphs and contour trees based on ExTreeM~\cite{ExtreeM}: 
(1) constraints for extremum graph and contour tree preservation (Section~\ref{sec:problem}), 
(2) our iterative correction algorithm (Section~\ref{subsec:correction_algorithm}), 
(3) reformulated constraints for distributed correction (Section~\ref{subsec:distributed}), and 
(4) theoretical upper bound on the number of correction iterations (Section~\ref{sec:convergence}).

\subsection{Constraints for Topological Preservation}
\label{sec:problem}

\updated{Our goal is to ensure that the extremum graph and contour tree of the corrected field $g$ are identical to those of the original field $f$. Specifically, this requires: (1) every critical point in $f$ remains a critical point of the same type at the same spatial location in $g$, with no new critical points introduced; (2) the extremum graph of $g$ has the same set of saddle--extremum edges as that of $f$; and (3) the contour tree of $g$ has the same set of arcs as that of $f$. To achieve this goal, we derive a set of sufficient (but not necessary) conditions based on ExTreeM~\cite{ExtreeM}.}
Given a user-defined absolute error bound $\xi$ (i.e., $|f_i - \hat{f}_i| \le \xi$), the contour tree of the decompressed field $\hat{f}$ matches that of the original field if the following structural constraints are satisfied.

\noindent\textbf{C1: Extremum Graph Constraints (EG Constraints).}
Since the extremum graph is the basis for contour tree construction in ExTreeM, it must be preserved, including both nodes and edges. The constraints are:
(1) \textbf{Critical point consistency.} The critical point type (minimum, maximum, or saddle) and spatial location of each \updated{critical point} in $\hat{f}$ must match those in $f$;
(2) \textbf{Saddle--extrema connectivity.} The connectivity between saddles and extrema must be preserved. 

\noindent\textbf{C2: Saddle Ordering Constraints (Saddle Constraints).}
To preserve the global ordering of merge and split events, the relative scalar ordering among all saddles must be identical between $\hat{f}$ and $f$. Specifically, for any two saddles $s_1$ and $s_2$, $\hat{f}_{s_1} < \hat{f}_{s_2}$ if and only if $f_{s_1} < f_{s_2}$.

\noindent\textbf{C3: Merge and Split Event Constraints (Event Constraints).}
To preserve the merge tree branches extracted by Extremum Graph Pairing (Step 2 of ExTreeM), the local scalar ordering among extrema and saddles involved in each topological event must remain unchanged. Specifically:
(1) \updated{For each join saddle, the extremum with the highest scalar value among its connected extrema in the extremum graph must remain unchanged. The same applies to split saddles by applying the procedure to $-f$.}
(2) \updated{For each minimum, the connected saddle with the lowest scalar value must remain unchanged. The same applies to maxima by applying the procedure to $-f$.}

These constraints preserve the extremum graph structure and the EGP pairing process, which together determine the contour tree in ExTreeM. C1 preserves critical points and the min/max neighbors of each vertex, ensuring that the selection of steepest ascending and descending neighbors remains unchanged. As integral paths are uniquely determined by these steepest neighbors, this preserves all integral paths and the induced extremum graph, yielding the same extremum graph as in the original field.
C2 preserves the relative ordering among saddles, ensuring that EGP processes them in the same sequence. Combined with C3, which preserves all saddle--extrema pairing decisions in the EGP process, every extracted branch remains identical. 
Since both the join and split trees are preserved (by symmetric application of C1--C3), the contour tree is preserved.

\begin{figure*}[htb!]
    \centering
    \includegraphics[width=\linewidth]{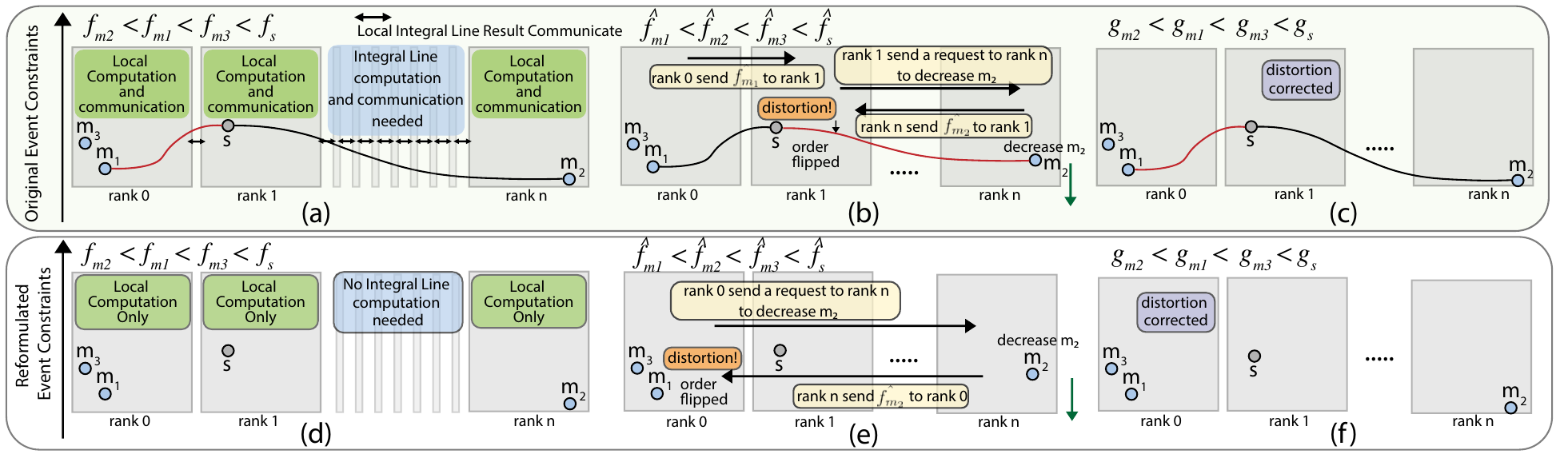} 
    \caption{Comparison of the original and reformulated event constraints. (a-c) Original event constraints. (d-f) Reformulated event constraints.}
    \label{fig:workflow}
    \vspace{-0.2in}
\end{figure*}

\subsection{Iterative Correction Algorithm}
\label{subsec:correction_algorithm}

We enforce the constraints in an iterative correction framework (Algorithm~\ref{alg:framework}), where the edited field $g$ is initialized with the decompressed data $\hat{f}$.
Each iteration examines the edited field $g$ to detect violations and identify vertices that need edits, and applies bounded monotonic edits to resolve the violations. Our framework supports either increasing- or decreasing-only edits; we use decreasing-only edits in this paper. \updated{Section~\ref{subsec:distributed} reformulates the event constraints for distributed execution.}

\begin{algorithm}[htb]
\caption{Iterative Workflow of EXaCTz}
\label{alg:framework}
\begin{algorithmic}[1]
\renewcommand{\algorithmicrequire}{\textbf{Input:}}
\renewcommand{\algorithmicensure}{\textbf{Output:}}

\Require Original field $f$, decompressed field $\hat{f}$, error bound $\xi$
\Ensure Set of applied edits $\mathcal{E}$

\State $g \gets \hat{f}$ \Comment{Initialize edited field to decompressed data}
\State $\mathcal{E} \gets \emptyset$ \Comment{Initialize the set of edits}

\While{\textbf{true}}
    \State $\mathcal{S} \gets \mathrm{CheckConstraints}(g, f)$ \Comment{Detect violations in EG, saddle-ordering, and event constraints}
    
    \If{$\mathcal{S} = \emptyset$}
        \State \textbf{break} \Comment{No violations found, exit loop}
    \EndIf

    \State $(g, \Delta \mathcal{E}) \gets \mathrm{ApplyBoundedEdits}(g, \mathcal{S}, \xi)$
    \State $\mathcal{E} \gets \mathcal{E} \cup \Delta \mathcal{E}$ \Comment{Accumulate new edits}
\EndWhile

\State \textbf{return} $\mathcal{E}$

\end{algorithmic}

\end{algorithm}

\noindent\textbf{Extremum Graph Constraints }require preserving critical point identities and the connectivity between saddles and extrema. 
Since integral paths are determined by local steepest directions, connectivity can be maintained by preserving local scalar orderings.
For each vertex $i$, let $N_{\max}(i)$ and $N_{\min}(i)$ denote the largest and smallest neighbors in the original field $f$, and let $\hat{N}_{\max}(i)$ and $\hat{N}_{\min}(i)$ denote those in the edited field $g$. For extrema, each local maximum (or minimum) must remain larger (or smaller) than all its neighbors.
A violation occurs if $\hat{N}_{\max}(i) \ne N_{\max}(i)$ or $\hat{N}_{\min}(i) \ne N_{\min}(i)$. 
If $\hat{N}_{\max}(i) \ne N_{\max}(i)$, we decrease the value of $\hat{N}_{\max}(i)$; 
if $\hat{N}_{\min}(i) \ne N_{\min}(i)$, we decrease the value of ${N}_{\min}(i)$.
For saddles, preserving only the steepest neighbors is insufficient, as saddles are where multiple components merge or split. We impose an additional constraint: the relative scalar ordering between a saddle and all its neighbors must be preserved, \updated{i.e., neighbors above (or below) the saddle in $f$ must remain above (or below) in $g$.}

\noindent\textbf{Saddle Ordering Constraints} require preserving the scalar ordering among saddle vertices. We compute saddle ordering in $f$ as the reference ordering. 
During each iteration, for a saddle $i$, we examine $i$'s immediate predecessor and successor in the reference ordering. A violation is detected if $g_i$ no longer lies between those of its adjacent saddles in the sorted order.
Specifically, for two adjacent saddles $i$ and $j$ in the reference ordering, if $f_i > f_j$ but $g_i < g_j$, we decrease $g_j$. Conversely, if $f_i < f_j$ but $g_i > g_j$, we decrease $g_i$.

\noindent\textbf{Merge and Split Event Constraints} require preserving branches extracted by Extremum Graph Pairing (EGP). 
In $f$, for merge events, a branch is defined by a minimum--saddle pair satisfying: for a minimum $i$, let $s$ be the \updated{lowest} saddle among all saddles connected to $i$; 
the pair $\langle i, s \rangle$ is extracted if $i$ is the \updated{highest} minimum among all minima connected to $s$. Split branches are defined symmetrically.
Due to compression errors, the minimum selected by EGP for a saddle may change in $g$. 
For a saddle $s$, let $m_1$ be the correct minimum in $f$ and $m_2$ be the minimum selected in $g$. To correct this violation, we decrease $g_{m_2}$ to enforce $g_{m_2} < g_{m_1}$.
The highest (lowest) saddle connected to each extremum is implicitly preserved by C2, which maintains the scalar ordering among saddles. The same argument applies to split events.

\subsection{Constraints for Distributed Iterative Correction}
\label{subsec:distributed}

To handle large-scale datasets in distributed settings (e.g., in situ compression pipelines), where data is distributed across processes, we reformulate the event  constraints to enable efficient distributed execution.
The original event constraints are not suitable for distributed environments as explicit integral path tracing is needed to identify the extrema connected to each saddle. Integral paths frequently cross partition boundaries (Figure~\ref{fig:workflow}(a)), leading to heavy inter-process communications and a severe scalability bottleneck. To eliminate the scalability bottleneck caused by integral path tracing, the reformulated event constraints preserve the ordering of all critical points, as detailed below.

\noindent\textbf{Reformulation of Event Constraints.}
The reformulated event constraints rely on the properties established by the EG constraints. Since EG constraints preserve both critical point classification and saddle--extrema connectivity locally, enforcing the ordering of all critical points guarantees that the relative ordering of the connected extrema of each saddle remains unchanged. As a result, the valid edges extracted by the EGP process are preserved, which satisfies the original event constraints. As illustrated in Figure~\ref{fig:workflow}(d-f), the reformulated event constraints avoid explicit cross-boundary integral path computations. Instead, each iteration only requires exchanging the scalar value of each critical point with its immediate predecessor and successor (determined by the sorted sequence in the original data). Violations are detected by comparing the current ordering with the original one (Figure~\ref{fig:workflow}(e)), and corrected by applying edits to restore the correct order (Figure~\ref{fig:workflow}(f)). While the reformulated event constraints introduce more edits, our experiments later demonstrate that the overhead is slight in practice.

\noindent\textbf{Distributed Parallel Implementation.} 
We partition the global domain into subdomains with ghost layers for cross-boundary connections. During each iteration, processes locally correct violations and synchronize ghost layers. To avoid race conditions, we employ a consistent rule during stencil exchanges (e.g., prioritizing smaller scalar modifications) to ensure deterministic convergence. The process repeats until no violations remain across the entire domain.

\subsection{Iteration Bound for Edit-Based Topology Correction}

\label{sec:convergence}
We establish an upper bound on the number of iterations for EXaCTz, stated below. Here, $\mathcal{D}_{\max}(G_R)$ represents the maximum path length in the reduced vulnerability graph $G_R$, defined later in this section.

\begin{theorem}\label{thm:convergence}
Given a maximum of $N$ allowed edits per vertex, the EXaCTz iterative correction algorithm converges within at most $N \cdot \mathcal{D}_{\max}(G_R)$ iterations. 
\end{theorem}

\begin{figure}
    \centering
    \includegraphics[width=\linewidth]{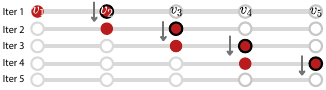}
    \caption{Cascading effect in a 1D sequence. Red circles denote flip pairs (i.e., adjacent vertices with inverted scalar ordering) at each iteration. }
    \label{fig:cast}
    \vspace{-0.1in}
\end{figure}

\noindent\textbf{Intuition: cascading effects in iterative correction.} Consider a cascading sequence to analyze the worst-case correction iterations. 
As an example, a 1D sequence of five vertices (Figure~\ref{fig:cast}) shows that resolving one flipped pair ($v_1$ and $v_2$) may trigger subsequent flips ($v_2$ and $v_3$), causing corrections to propagate along the sequence. While this propagation is linear in 1D, higher-dimensional datasets introduce more complex dependencies. To bound the maximum number of iterations, we model such cascades as paths in a graph, as described below.

\noindent\textbf{Vulnerability graph construction.} 
To analyze cascading effects in the correction process, we construct a sequence of graphs. Let $V$ be the set of all vertices in the scalar field, $f$ and $\hat{f}$ denote the original and decompressed data, and $\xi$ denote the prescribed absolute error bound.

\begin{figure*}[htb]
    \centering
    \includegraphics[width=\linewidth]{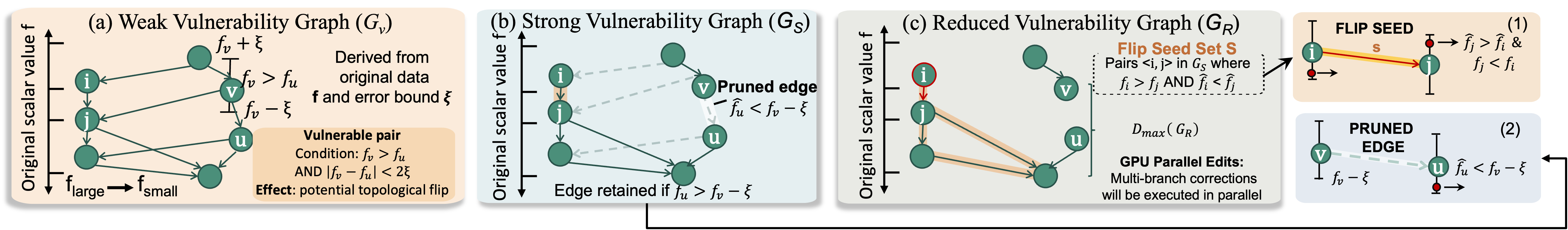}
    \vspace{-5mm}
    \caption{Illustration of the data-aware convergence bound for EXaCTz. 
    (a) The weak vulnerability graph $G_V$ captures all potential order flips based solely on the prescribed error bound. 
    (b) The strong vulnerability graph $G_S$ that prunes edges using the decompressed data. 
    (c) The reduced vulnerability graph $G_R$ that traces cascading paths exclusively from actual flip seeds $S$.}
    \label{fig:convergence_bound}
    \vspace{-0.2in}
\end{figure*}

Given the input data $f$, we first identify all potential topological violations based solely on the prescribed error bound $\xi$. We define a \textit{weak vulnerability pair} as adjacent vertices $(u,v)$ satisfying $|f_u - f_v| \le 2\xi$, which is the prerequisite for any order flip (i.e., scalar ordering violations that break our preservation constraints either between adjacent neighbors or across the sorted sequence of critical points). Adding a directed edge $u \to v$ for every such pair (where $f_u > f_v$) forms the \textbf{\textit{Weak Vulnerability Graph}} $G_V = (V, E_V)$, as shown in Figure~\ref{fig:convergence_bound}(a). Because all edges point from larger to smaller original values, $G_V$ is inherently a Directed Acyclic Graph, as any directed cycle would imply a decreasing chain of values that returns to its starting point.
Next, we prune $G_V$ using the decompressed data $\hat{f}$. Because our monotonic edit strategy only decreases scalar values, a vertex $u$ can never drop below its absolute minimum $f_u - \xi$. If $\hat{f}_v < f_u - \xi$, $u$ will always remain greater than $v$, making the edge unconditionally safe. We prune these safe edges and retain only the \textit{strong vulnerability pairs} where $\hat{f}_v \ge f_u - \xi$, forming the \textbf{\textit{Strong Vulnerability Graph}} $G_S$, as shown in Figure~\ref{fig:convergence_bound}(b).
We then extract the \textbf{\textit{Reduced Vulnerability Graph}} $G_R$ to capture the exact scope of the corrections. Corrections in EXaCTz originate exclusively from actual flip seeds, i.e., vertex pairs whose scalar-value ordering is inverted in the decompressed data. We define the flip seed set as $S = {(u,v) \in E_S \mid \hat{f}_u \le \hat{f}_v}$.
By tracing propagation paths from flip seeds within $G_S$, we extract $G_R$, as shown in Figure~\ref{fig:convergence_bound}(c). $\mathcal{D}_{\max}(G_R)$ is defined as the length of the longest directed path in $G_R$.

\subsubsection{Proof of Bounded Convergence}

Following the monotonic edit strategy, each vertex is decreased by at most one step $\Delta = \xi / N$ per iteration, requiring at most $N$ iterations to reach the lower bound $f_v - \xi$.

\noindent\textit{Proof.}
We prove by induction on the depth of vertices in $G_R$.
Define the depth $\mathrm{dep}(v)$ as the maximum number of vertices along any directed path from a seed pair endpoint to $v$. Then $\max_v \mathrm{dep}(v) = \mathcal{D}_{\max}(G_R)$.
\textbf{Base case ($\mathrm{dep}(v)=1$).}
Vertices with $dep(v)=1$ correspond to endpoints of violation pairs in the seed set. Since each vertex can be edited at most $N$ times and at most once per iteration, all such vertices stabilize within $N$ iterations.
\textbf{Inductive step.}
Assume all vertices with depth at most $d-1$ stabilize within $N(d-1)$ iterations.
Consider a vertex $v$ with $\mathrm{dep}(v)=d$. By the observation above, any new violation involving $v$ can only be introduced by its predecessors. Since all predecessors have depth at most $d-1$, they stabilize by iteration $N(d-1)$.
After that point, no new violations involving $v$ can be introduced. Since $v$ can be edited at most $N$ times, it stabilizes within an additional $N$ iterations.
Thus, all vertices with depth at most $d$ stabilize within $Nd$ iterations.
Setting $d = \mathcal{D}_{\max}(G_R)$ completes the proof. \qed

\noindent\textbf{Comparing Theoretical Max Iteration \updated{vs.} Actual Iteration.}
While the established theoretical upper bound guarantees convergence, worst-case cascading edits are exceedingly rare in practice. To provide context for the theoretical iteration limit, Table~\ref{tab:topo_metrics} compares the theoretical vs.\ actual iterations on real-world datasets. Specifically, we report four metrics representing the proportion of affected vertices relative to the total vertex count $|V|$ (all reported as percentages): 
$|G_V|$ ($|G_V|/|V|$), the fraction of vertices in the weak vulnerability graph; 
$|G_S|$ ($|G_S|/|V|$), the fraction in the strong vulnerability graph; 
$|G_R|$ ($|G_R|/|V|$), the fraction in the reduced vulnerability graph; and 
\textit{Edit}, the fraction of vertices actually edited by EXaCTz.
We then compare the empirical number of correction iterations (\textbf{Actual Iter.}) against the worst-case bound (\textbf{Theo. Max Iter.}) established in Theorem~\ref{thm:convergence}.

\begin{table}[htb]
\small
\centering
\caption{Iteration numbers for EXaCTz across different datasets under a relative error bound $\xi = 10^{-4}$ (normalized to the data range, using ZFP as base compressor). All ratios are in percentages (\%).}
\vspace{-2mm}
\label{tab:topo_metrics}
\renewcommand{\arraystretch}{1.2}
\setlength{\tabcolsep}{2.3pt}
\begin{tabular}{lccccccc}
\toprule
Dataset & Dimension &  \begin{tabular}{c}$|G_V|$\\(\%)\end{tabular}& \begin{tabular}{c}$|G_S|$\\(\%)\end{tabular}& \begin{tabular}{c}$|G_R|$\\(\%)\end{tabular}& \begin{tabular}{c}Edit\\(\%)\end{tabular} &  \begin{tabular}{c}\textbf{Theo. Max}\\\textbf{Iter.}\end{tabular} & \begin{tabular}{c}\textbf{Actual}\\\textbf{Iter.}\end{tabular} \\
\midrule
QMCPack    & $69\times69\times115$ & 23.69 & 13.41  & 1.31 & 1.40 & 275 & 71 \\
AT         & $177\times95\times48$ & 5.78  & 3.00 & 0.02 & 0.01  & 85 & 3 \\
Vortex     & $128^3$               & 12.43  & 7.16 & 0.43 & 0.75  & 735 & 160 \\
Turbulence & $256^3$               & 17.22  & 11.50 &  4.70 & 4.28  & 1015 & 169 \\
NYX ($512^3$) & $512^3$            & 35.20 & 20.93 & 0.76 & 0.73 & 2805 & 25 \\
Combustion & $560^3$               & 38.40 & 24.49 & 0.44 & 0.31 & 3025 & 56 \\
\bottomrule
\end{tabular}
\vspace{-2mm}
\end{table}

As detailed in Table~\ref{tab:topo_metrics}, the actual iterations required in practice remain up to \updated{two orders of magnitude} lower than the worst-case theoretical limits. While a significant fraction of vertices initially appear at risk (weak vulnerability), the reduced vulnerability graph effectively prunes unconditionally safe edges. Furthermore, actual topological inversions (flips) in the decompressed data are exceedingly rare, resulting in a low final edit ratio. This result suggests that the theoretical bound can serve as a conservative guide for selecting error bounds, and we leave its practical use as future work.

\section{Evaluation}

We evaluate EXaCTz (our method) from three perspectives: (1) single-node performance, including comparisons with TopoA~\cite{gorski2025general} using a serial CPU implementation, single-GPU performance, and the trade-offs introduced by the reformulated event constraints; (2) scalability of distributed correction; and (3) preservation of the extremum graph and contour tree under a prescribed error bound.

\subsection{Experiment Setup}
All experiments are conducted on the Perlmutter supercomputer at NERSC, using nodes with four NVIDIA A100 GPUs (40~GB VRAM each). The largest dataset that EXaCTz can process on a single GPU is 5~GB, as its peak GPU memory footprint is approximately 8$\times$ the raw data size due to the simultaneous storage of the original and decompressed fields, critical point indices, sorted orders, and extremum graph connectivity. 

Our implementation uses CUDA and CUDA-aware MPI for direct GPU-to-GPU communications. 
\updated{We encode edits as a per-vertex delta map of integers in $[0, N+1]$, where $0$ indicates no edit, $N+1$ flags a lossless edit, and $N$ is small in practice (e.g., $N=5$; see Section~\ref{sec:bg_monotonic_edit}). Both the delta map and lossless edit values are compressed using ZSTD.}

\noindent\textbf{Datasets.}~We use a variety of real-world scientific datasets from cosmology, combustion, and climate applications, with sizes ranging from 4~MB to 512~GB, in both single-node and distributed environments. 

\noindent\textbf{Baseline.} 
We compare against TopoA~\cite{gorski2025general}, a state-of-the-art contour-tree-preserving baseline that outperforms prior methods such as TopoSz~\cite{toposz} in efficiency. For fair comparison, a serial CPU version of EXaCTz is implemented and evaluated against the serial CPU implementation of TopoA.
Additionally, we include pMSz~\cite{li_msz} (a method for preserving Morse--Smale segmentations) for quality comparison on extremum graphs and contour trees.

\noindent\updated{\textbf{Base Compressor Configuration.} We evaluate EXaCTz 
using five lossy compressors: two CPU-based (SZ3, ZFP) and three 
GPU-accelerated (cuZFP, MGARD-GPU~\cite{MGARD}, cuSZp). As cuZFP 
lacks strict pointwise error bounds, we use the actual maximum 
absolute error as the effective $\xi$. Tables~\ref{tab:single_node} 
and~\ref{tab:multi_node} evaluate all five; other experiments use 
SZ3 as baseline due to its high compression ratio, which creates more challenging correction cases. Additional results (parameter study on $N$, evaluations with TTHRESH~\cite{TTHRESH} and 
SPERR~\cite{SPERR}) are in Appendices~B and~C. Unless specified, we use 
a fixed relative error bound $\xi = 10^{-4}$ (normalized to the 
data range). EXaCTz achieves 1.00 precision across all metrics and datasets (Appendix~F).}

\noindent\textbf{Evaluation Metrics.} 
We use three metric categories for evaluation: 

(1) \textit{Compression Efficiency}: We report the \textit{Compression Ratio (CR)} (original data size over compressed size), \textit{Overall Compression Ratio (OCR)} (original size over compressed size plus correction edits), and the \textit{Edit Ratio} (fraction of edited data points by EXaCTz). 

(2) \textit{Performance and Scalability}: \updated{We report Overall Throughput (OT), defined as the original data size 
divided by total correction time. For single-GPU settings, OT reflects 
correction time only; for multi-GPU settings (Table~\ref{tab:multi_node}), 
OT includes communication and synchronization overhead.} We report \textit{strong scaling efficiency}, $E(p) = \frac{T(1)}{p\,T(p)}$, and \textit{weak scaling efficiency}, defined as $W(p) = \frac{T(1)}{T(p)}$, \updated{where $p$ is the number of GPUs,} which measures the ability to maintain 
constant runtime as the dataset size scales proportionally with the number of GPUs.

(3) \textit{Topological Feature Preservation}: Existing topology metrics (e.g., Wasserstein distance~\cite{mileyko2011probability}) do not directly capture the exact structural consistency required by EXaCTz. Therefore, we quantify structural consistency between $f$ and $\hat{f}$ at three levels using recall metrics: \textit{CP-Recall} compares the sets of critical points (a match requires both the spatial location and type to remain unchanged); \textit{EG-Recall} evaluates the fraction of saddle-extremum edges that are preserved; and \textit{CT-Recall} evaluates the fraction of merge and split arcs that are preserved. \updated{Precision metrics are reported in Appendix~F.}

\subsection{Single-Node Evaluation}
We evaluate EXaCTz to characterize both algorithmic efficiency under serial CPU execution and performance on a single GPU.
\subsubsection{Comparison with TopoA under serial CPU implementation}\label{sec:serial_cpu}
\begin{mdframed}
\noindent\textbf{Observation 1.} In a serial CPU setting, EXaCTz runs up to \textbf{213$\times$} faster and delivers up to \textbf{2.65$\times$} higher OCR than TopoA.
\end{mdframed} 
\updated{To ensure a fair comparison with TopoA, we (1) add a topology 
simplification step ($\epsilon = 10^{-5}$) using TTK~\cite{ttk}, 
as TopoA requires simplification to complete within practical time 
limits, and (2) implement a serial CPU version of EXaCTz to match 
TopoA's execution environment. Both fields are simplified using the 
same threshold, and the contour tree from the simplified original 
field serves as the reference. All other settings are identical 
($\xi = 10^{-4}$).}
\begin{figure}[htb]
\centering
\includegraphics[width=\linewidth]{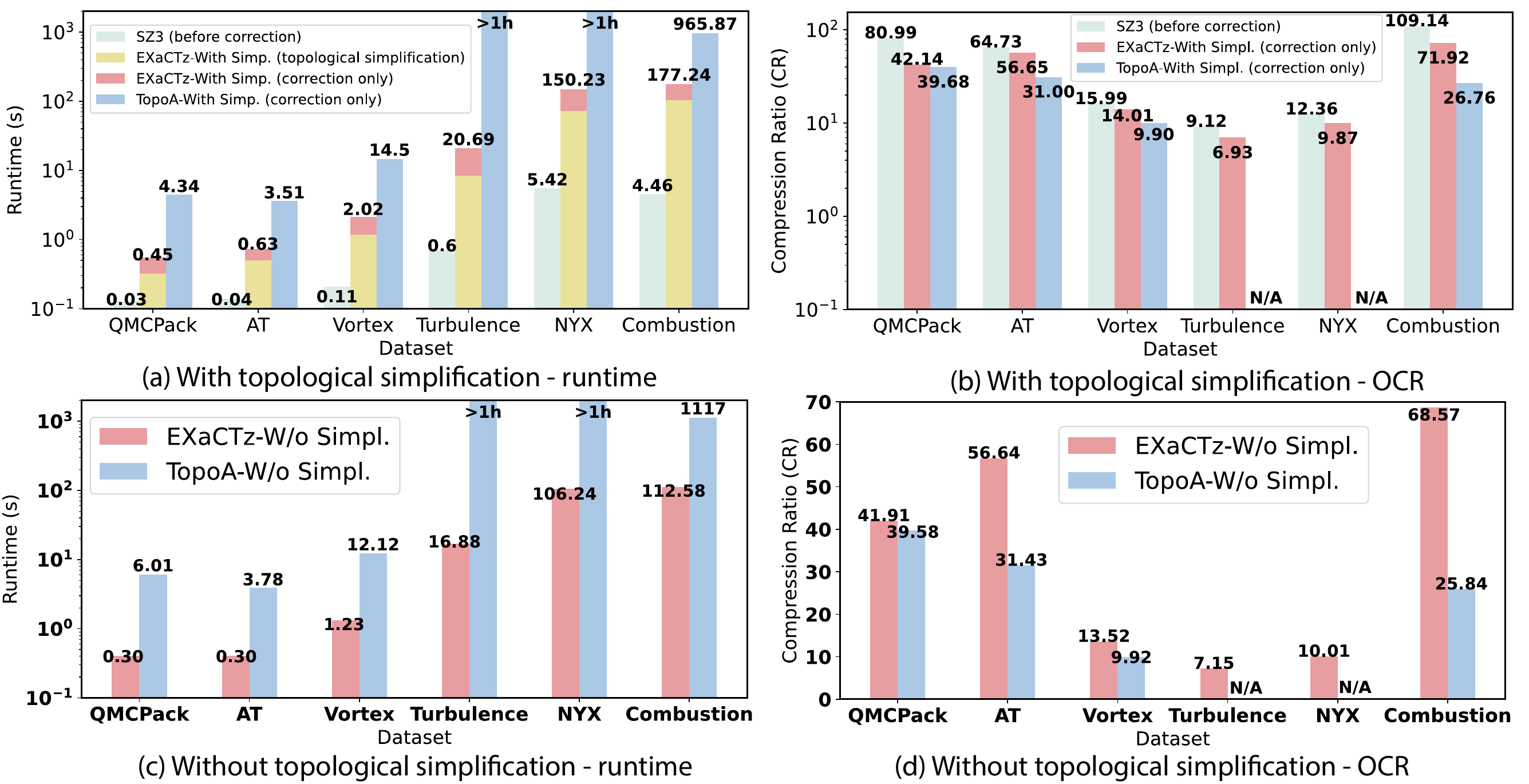}
\vspace{-4mm}
\caption{Serial CPU comparison: TopoA vs. EXaCTz (SZ3, $\xi = 10^{-4}$). (a) Runtime (log scale) and (b) OCR with topological simplification ($\epsilon = 10^{-5}$); (c, d) corresponding metrics without simplification. \updated{Legend: green = SZ3 before correction, yellow = EXaCTz 
topological simplification, red = EXaCTz correction only, blue = TopoA correction only.} N/A indicates out-of-time execution ($> 1$ hour).}
\label{fig:topoA_comparison}
\end{figure}

\noindent\textbf{Runtime.} 
As shown in Figure~\ref{fig:topoA_comparison}, EXaCTz achieves speedups of up to \textbf{174$\times$} and \textbf{213$\times$} over TopoA, with and without topological simplification, respectively. While both methods complete successfully on smaller datasets like QMCPack and Adenine Thymine (AT), TopoA's runtime increases rapidly with data size due to explicit contour tree reconstruction overhead. \updated{On the $256^3$ turbulence dataset, EXaCTz finishes in 
20.69s with simplification (including 8.31s for the simplification 
step) and 16.88s without, whereas TopoA fails to converge within 
one hour.} This results in a 213$\times$ speedup relative to the one-hour time limit that TopoA exceeded.
We also observe that persistence-based topological simplification incurs non-negligible cost, especially on large datasets, and is difficult to scale as the data size grows. This further motivates preservation strategies that directly enforce full contour tree consistency, thereby improving scalability.

\noindent\textbf{Overall Compression Ratio.} 
From the compression perspective, EXaCTz achieves up to 2.4$\times$ and 2.65$\times$ higher OCRs than TopoA, with and without topological simplification, respectively, as shown in Figure~\ref{fig:topoA_comparison} (b) and (d).

\subsubsection{Single-GPU Performance}
\begin{mdframed}
\noindent\textbf{Observation 2.} In a single-GPU configuration, EXaCTz completes topology correction for all datasets within 1.6 seconds, achieving GB/s-scale throughput (e.g., 3.85 GB/s on the SZ3-compressed Combustion dataset), compared to MB/s-scale throughput of TopoA ($\sim$1 MB/s).
\end{mdframed}
We evaluate the performance of EXaCTz on a single GPU under full contour tree preservation without topological simplification. 

\noindent\textbf{Runtime.} As shown in Table~\ref{tab:single_node}, EXaCTz typically completes the topological correction within 1.6 seconds on a single GPU (maximum 1.58s for cuZFP on NYX ($512^3$)). Compared to prior CPU-based topology-preserving methods, {\tool} reduces runtime by orders of magnitude, reaching up to $3285\times$ improvement over TopoA. Notably, correction time matches the order of magnitude of base compression time, even for fast GPU-native compressors. Larger datasets require slightly more time as complex structures induce more local distortions and additional iterations.

\noindent\textbf{Overall Throughput.} 
EXaCTz achieves a practical throughput, ensuring that topology correction does not become a severe bottleneck in the compression pipeline. For example, EXaCTz achieves a throughput of 3.85 GB/s on the Combustion dataset (SZ3 compressed), where TopoA only delivers a throughput of $\sim$1 MB/s. Compared with GPU-native compressors such as cuSZp, EXaCTz offers a practical trade-off between throughput, compression efficiency, and exact topological guarantees. As shown in Table~\ref{tab:single_node}, cuSZp achieves extreme speed by severely sacrificing compression ratio, whereas EXaCTz guarantees exact topology and maintains high overall CRs at practical speeds.

\noindent\textbf{Overall Compression Ratio.} Even with contour tree preservation, the OCR of EXaCTz remains practical. Because EXaCTz modifies only a sparse subset of data, the required edit ratio is generally below 15\% (with rare exceptions like cuSZp on NYX ($512^3$)).

\noindent\textbf{Impact of Base Compressors.} Different base compressors exhibit distinct performance trade-offs. SZ3 achieves higher OCRs (e.g., 42.58 on QMCPack) due to its higher base CR. ZFP and cuZFP tend to require more iterations (e.g., up to 224 and 183, respectively). This is because ZFP’s block-based transform tends to introduce more severe initial topological distortions. For cuZFP, the lack of strict pointwise error bounds can further amplify this effect. Meanwhile, cuSZp yields lower OCRs, which aligns with cuSZp's primary design objective of maximizing computational throughput rather than compression efficiency.

\begin{table}[htb]
\small
\centering
\caption{Single-GPU performance of EXaCTz on a single GPU. 
(SZ3, ZFP, cuSZp, MGARD-GPU: $\xi=10^{-4}$; cuZFP: target CR=10). 
Throughput is in GB/s. EXaCTz time reports correction only. 
SZ3 (32 threads) and ZFP are CPU-based; others are GPU-based.}
\label{tab:single_node}
\renewcommand{\arraystretch}{0.8} 
\setlength{\tabcolsep}{0.8pt} 
\begin{tabular}{l c c c c c | c c c c}
\toprule
\multirow{2}{*}{Dataset} & \multirow{2}{*}{\shortstack{Size\\(GB)}} & \multicolumn{4}{c|}{\shortstack{Stage 1 \\ Compression-only}} & \multicolumn{4}{c}{\shortstack{Stage 2 \\ EXaCTz Correction-only}} \\
\cmidrule(lr){3-6} \cmidrule(lr){7-10} 
& & Type & \begin{tabular}{c}Time\\(s)\end{tabular}  & \begin{tabular}{c}Thrpt.\\(GB/s)\end{tabular}& CR & \begin{tabular}{c}Time\\(s)\end{tabular} & \begin{tabular}{c}Overall\\Thrpt. (GB/s)\end{tabular} & Iter. & OCR \\
\midrule
\multirow{5}{*}{QMCPack} & \multirow{5}{*}{0.004} 
& SZ3 & 0.01 & 0.40 & 80.99 & 0.05 & 0.08 & 37 & 42.58 \\
& & ZFP & 0.09 & 0.04 & 16.17 & 0.03 & 0.13 & 71 & 15.17 \\
& & cuSZp & <0.01 & >0.40 & 7.90 & 0.27 & 0.01 & 16 & 6.75 \\
& & cuZFP & <0.01 & >0.40 & 9.09 & 0.56 & 0.01 & 115 & 6.72\\
& & MGARD-GPU & 0.09 & 0.04 & 3.60 & 0.07 & 0.06 & 36 & 3.50 \\
\midrule
\multirow{5}{*}{AT} & \multirow{5}{*}{0.006} 
& SZ3 & 0.02 & 0.30 & 64.73 & 0.01 & 0.60 & 7 & 56.59 \\
& & ZFP & 0.01 & 0.60 & 13.32 & 0.04 & 0.15 & 3 & 13.32 \\
& & cuSZp & 0.01 & 0.60 & 4.91 & 0.16 & 0.04 & 7 & 4.83 \\
& & cuZFP & 0.04 & 0.15 & 9.72 & 0.09 & 0.07 & 8 & 9.70 \\
& & MGARD-GPU & 0.08 & 0.07 & 5.02 & 0.05 & 0.12 & 5 & 4.94\\

\midrule
\multirow{5}{*}{Vortex} & \multirow{5}{*}{0.015} 
& SZ3 & 0.05 & 0.30 & 15.99 & 0.03 & 0.50 & 26 & 13.98 \\
& & ZFP & 0.04 & 0.38 & 8.02 & 0.10 & 0.15 & 160 & 7.78 \\
& & cuSZp & 0.01 & 1.50 & 5.64 & 0.23 & 0.07 & 21 & 5.39 \\
& & cuZFP & 0.01 & 1.50 & 9.99 & 1.13 & 0.01 & 131 & 8.70 \\
& & MGARD-GPU & 0.15 & 0.10 & 6.65 & 0.06 & 0.25 & 49 & 6.45 \\
\midrule
\multirow{5}{*}{Turbulence} & \multirow{5}{*}{0.125} 
& SZ3 & 0.38 & 0.33 & 9.12 & 0.14 & 0.89 & 24 & 7.15 \\
& & ZFP & 0.30 & 0.42 & 5.64 & 0.42 & 0.30 & 169 & 4.76 \\
& & cuSZp & 0.02 & 6.25 & 6.32 & 0.54 & 0.23 & 24 & 5.24 \\
& & cuZFP & 0.02 & 6.25 & 9.99 & 0.93 & 0.13 & 115 & 5.36 \\
& & MGARD-GPU & 0.41 & 0.30 & 4.98 & 0.27 & 0.46 & 122 & 4.34 \\
\midrule
\multirow{5}{*}{NYX ($512^3$)} & \multirow{5}{*}{1.000} 
& SZ3 & 3.21 & 0.31 & 12.36 & 0.37 & 2.70 & 25 & 9.97 \\
& & ZFP & 2.77 & 0.36 & 6.00 & 0.89 & 1.12 & 165 & 5.75 \\
& & cuSZp & 0.01 & 100.00 & 5.84 & 0.53 & 1.89 & 25 & 4.72 \\
& & cuZFP & 0.11 & 9.09 & 9.99 & 1.58 & 0.63 & 183 & 5.63 \\
& & MGARD-GPU & 1.57 & 0.64 & 10.53 & 0.55 & 1.82 & 64 & 9.09 \\
\midrule
\multirow{5}{*}{Combustion} & \multirow{5}{*}{1.310} 
& SZ3 & 3.32 & 0.39 & 109.14 & 0.34 & 3.85 & 56 & 70.48 \\
& & ZFP & 3.21 & 0.41 & 15.42 & 0.98 & 1.34 & 224 & 15.11 \\
& & cuSZp & 0.01 & 131.00 & 5.10 & 0.29 & 4.52 & 22 & 4.84 \\
& & cuZFP & 0.06 & 21.83 & 9.99 & 1.35 & 0.97 & 22 & 9.92 \\
& & MGARD-GPU & 2.72 & 0.48 & 19.32 & 0.56 & 2.34 & 103 & 16.82 \\
\bottomrule
\end{tabular}
\vspace{-0.1in}
\end{table}

\begin{figure}[htb]
    \centering
    \includegraphics[width=\linewidth]{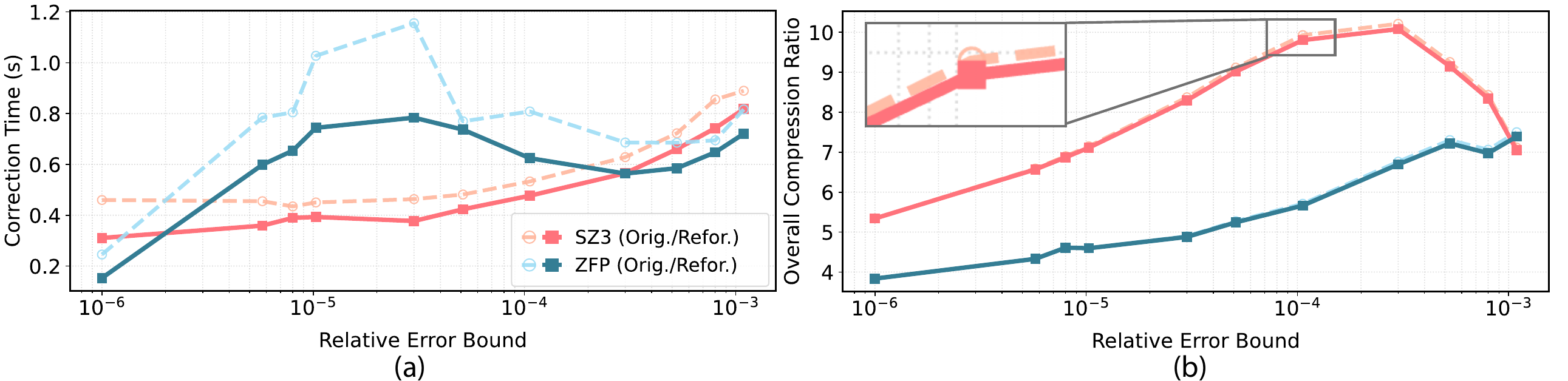}
    \vspace{-0.18in}
    \caption{Trade-offs of reformulated event constraints on NYX ($512^3$) across varying error bounds using SZ3 and ZFP: (a) correction time and (b) OCR. Red/blue lines: SZ3/ZFP. Solid/dashed lines: reformulated/original event constraints.}
  \label{fig:reformulation_tradeoff}
  \vspace{-0.1in}
\end{figure}


\subsubsection{Compute–Storage Trade-Offs with Event Constraints}

\begin{mdframed}
\noindent\textbf{Observation 3.}
Across the tested error bounds, the reformulated event constraints make the correction process 1.5$\times$ faster; the OCR drop remains modest relative to the runtime reduction.
\end{mdframed}
Figure~\ref{fig:reformulation_tradeoff} quantifies the performance impact of reformulating event constraints on the NYX dataset ($512^3$) across relative error bounds from $10^{-6}$ to $10^{-3}$. By replacing explicit integral path tracing with critical point scalar ordering constraints, the reformulated event constraints consistently accelerate the correction process. As shown in Figure~\ref{fig:reformulation_tradeoff}(a), it achieves up to a 1.5$\times$ speedup in correction time on a single GPU for SZ3 and ZFP. Although this acceleration may trade off storage efficiency via additional edits, Figure~\ref{fig:reformulation_tradeoff}(b) shows the OCR decrease remains small across compressors. In practice, this reformulation is a user-configurable option, enabling domain scientists to balance computational efficiency and storage overhead.

\subsection{Distributed and Multi-GPU Evaluation}

\begin{mdframed}
\noindent\textbf{Observation 4.}
The reformulated event constraints improve parallel efficiency from 6.4\% to 55.6\% at 128 GPUs, enabling correction for datasets up to 512~GB in under 48 seconds (with an aggregate throughput of up to 32.69~GB/s) in distributed settings.
\end{mdframed}
We evaluate distributed multi-GPU scalability and parallel efficiency. Strong scaling is evaluated on 1 to 16 GPUs using a fixed $1024^3$ global problem size. Weak scaling is performed on up to 128 GPUs (fixed $512^3$ subdomain per GPU), and includes an analysis of the reformulated event constraints. We also report performance on larger datasets. 

\noindent\textbf{Strong Scaling.} We evaluate strong scaling by fixing the global input size to a $1024^3$ scalar field and increasing the number of GPUs from 1 to 16. As shown in Figure~\ref{fig:strong}, the strong scaling efficiency remains close to ideal at small scales (99.5\% at 2 GPUs and 92.5\% at 4 GPUs), and gradually decreases to 79.1\% at 8 GPUs and 64.1\% at 16 GPUs. The efficiency drop is expected under strong scaling. When the global problem size is fixed, increasing the number of GPUs reduces the amount of local correction work assigned to each GPU. However, boundary communications, including ghost-layer exchange and critical-point value synchronization, does not decrease proportionally. As a result, the computation time per GPU shrinks faster than the communication time, increasing the relative impact of communication overhead and lowering parallel efficiency at larger GPU counts.

\begin{figure}[htb]
    \centering
    \includegraphics[width=\linewidth]{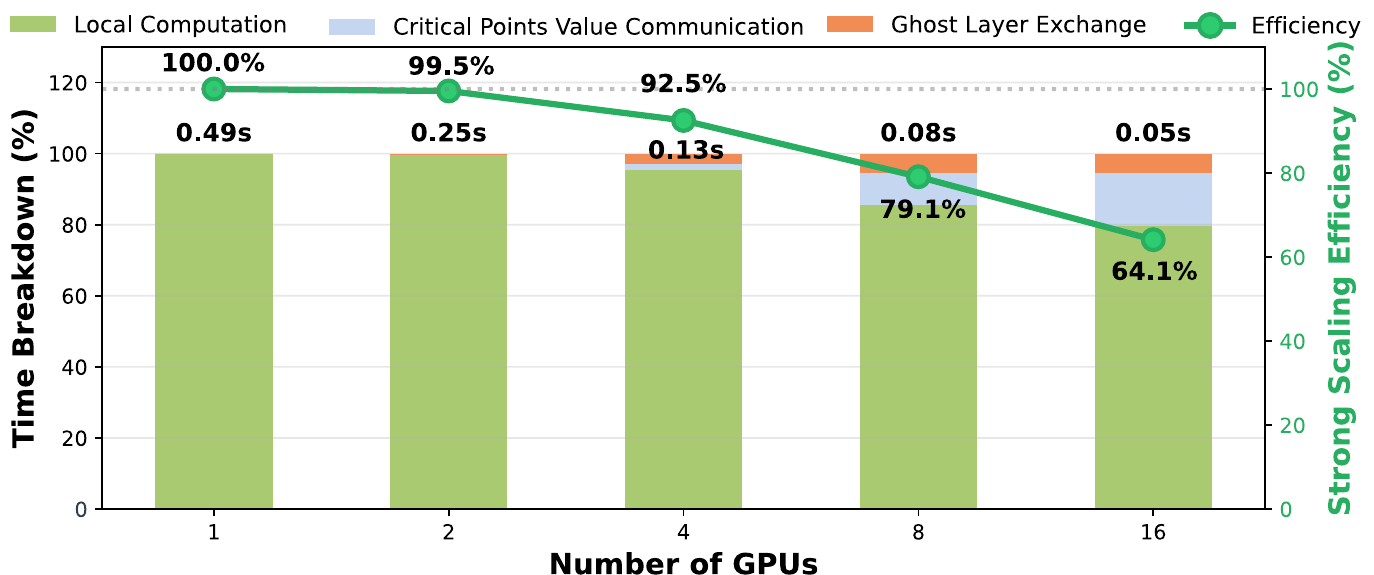}
    \vspace{-4mm}
    \caption{Strong scaling results of EXaCTz on a fixed-size $1024^3$ NYX dataset. The stacked bars show the percentage time breakdown of the execution, and the line shows the corresponding strong scaling efficiency.}
    \label{fig:strong}
    \vspace{-0.05in}
\end{figure}

\noindent\textbf{Weak Scaling.} The weak scaling behavior of EXaCTz is evaluated as the system scales up to 128 GPUs with a fixed $512^3$ subdomain per GPU. 
We generate larger datasets by upsampling a $512^3$ NYX volume. The upsampling is designed to yield similar compression distortions across blocks, as the effective problem complexity is more closely tied to topology distortion than to resolution alone. \updated{We additionally report weak scaling on downscaled real data 
in Appendix~I.}
\begin{figure}[htb!]
    \centering
    \includegraphics[width=\linewidth]{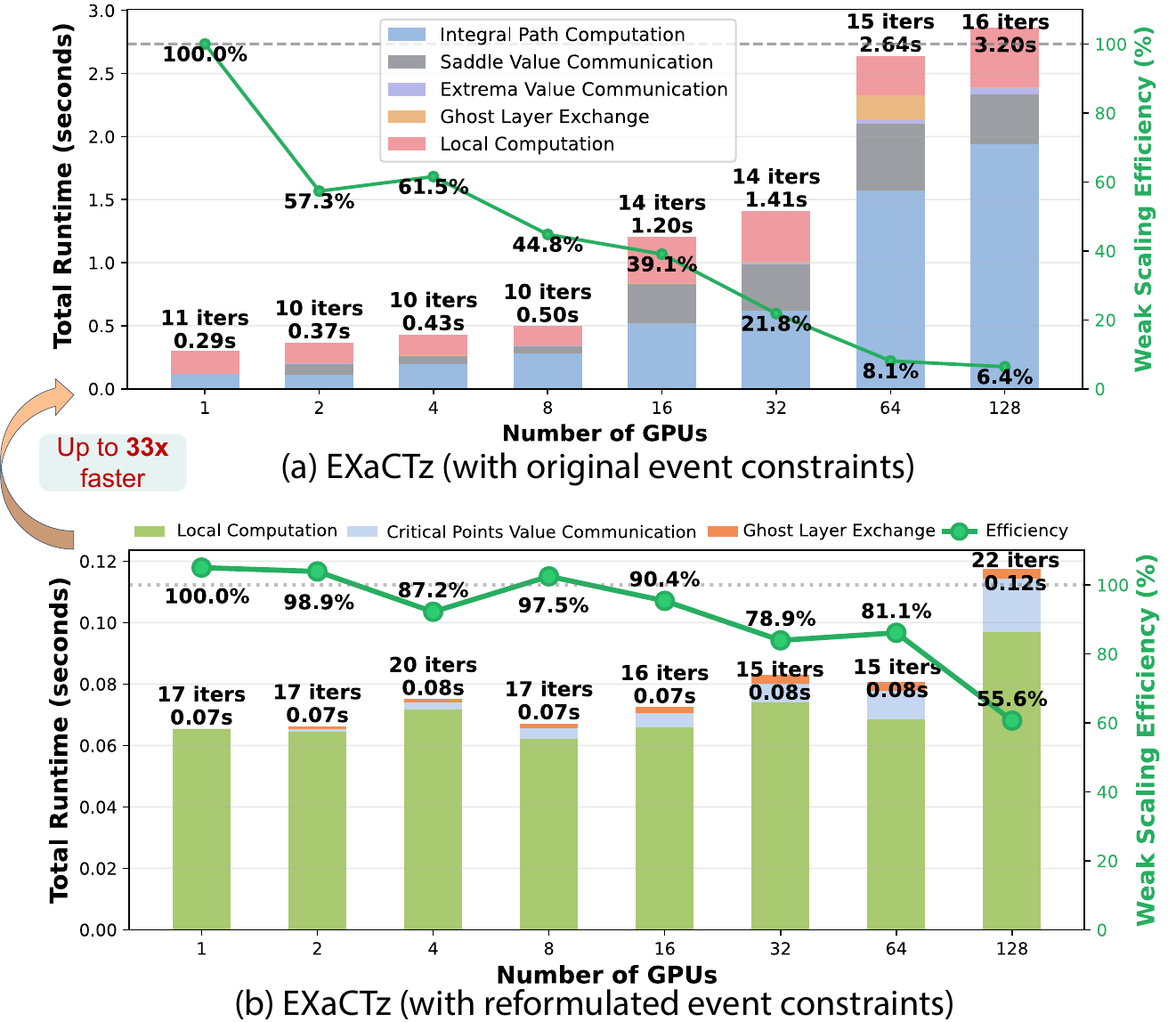}
    \vspace{-6mm}
    \caption{Weak Scaling Performance of EXaCTz from 1 to 128 GPUs with (a) original event constraints. (b) reformulated event constraints.}
    \label{fig:weak}
\end{figure}

\begin{table}[htb!]
\small
\centering
\caption{\updated{Multi-GPU performance evaluation of EXaCTz. \#GPUs denotes the minimum number of GPUs required to fit each dataset. Compression throughput is computed per GPU block (4 GB/GPU); correction throughput is computed over the entire dataset.}}
\vspace{-2mm}
\label{tab:multi_node}
\renewcommand{\arraystretch}{0.7} 
\setlength{\tabcolsep}{0.5pt} 
\begin{tabular}{lcccccc|cccc}
\toprule
\multirow{2}{*}{Dataset} & \multirow{2}{*}{\shortstack{Size\\(GB)}} & \multirow{2}{*}{\# GPUs} & \multicolumn{4}{c|}{\shortstack{Stage 1 \\ Compression-only}} & \multicolumn{4}{c}{\shortstack{Stage 2 \\ EXaCTz Correction-only}} \\
\cmidrule(lr){4-7} \cmidrule(lr){8-11}
& & & Type & Time (s) & \shortstack{Thrpt.\\(GB/s)} & CR & Time (s) & \shortstack{Thrpt.\\(GB/s)} & Iter. & OCR \\
\midrule
\multirow{5}{*}{\begin{tabular}{c}NYX\\ ($1024^3$)\end{tabular}} & \multirow{5}{*}{8} & \multirow{5}{*}{2} 
& SZ3 & 13.11 & 0.31 & 12.25 & 3.18 & 2.52 & 71 & 10.29 \\
& & & ZFP & 13.07 & 0.31 & 5.97 & 10.53 & 0.76 & 149 & 5.37 \\
& & & cuSZp & 0.03 & 133.33 & 4.85 & 3.54 & 2.26 & 87 & 4.52 \\
& & & cuZFP & 0.18 & 22.22 & 9.61 & 10.67 & 0.75 & 166 & 6.71 \\
& & & MGARD-GPU & 7.08 & 0.56 & 8.11 & 9.10 & 0.88 & 251 & 7.31 \\
\midrule
\multirow{5}{*}{\begin{tabular}{c}NYX\\($2048^3$)\end{tabular}} & \multirow{5}{*}{64} & \multirow{5}{*}{16} 
& SZ3 & 16.14 & 0.25 & 12.59 & 11.48 & 5.57 & 101 & 10.34 \\
& & & ZFP & 11.07 & 0.36 & 6.10 & 24.07 & 2.66 & 412 & 5.85 \\
& & & cuSZp & 0.03 & 133.33 & 4.77 & 9.98 & 6.41 & 96 & 4.46 \\
& & & cuZFP & 0.11 & 36.36 & 9.20 & 15.23 & 4.20 & 194 & 5.56 \\
& & & MGARD-GPU & 7.06 & 0.57 & 8.15 & 23.05 & 2.78 & 335 & 7.37 \\
\midrule
\multirow{5}{*}{\begin{tabular}{c}Turbulence\\ ($4096^3$)\end{tabular}} & \multirow{5}{*}{512} & \multirow{5}{*}{128} 
& SZ3 & 16.50 & 0.24 & 33.72 & 15.66 & 32.69 & 241 & 25.21 \\
& & & ZFP & 12.74 & 0.31 & 10.40 & 35.78 & 14.31 & 441 & 10.11 \\
& & & cuSZp & 0.03 & 133.33 & 6.12 & 18.22 & 28.10 & 111 & 5.72 \\
& & & cuZFP & 0.23 & 17.39 & 9.99 & 47.56 & 10.77 & 612 & 5.10 \\
& & & MGARD-GPU & 8.76 & 0.46 & 10.24 & 40.74 & 12.57 & 159 & 9.37 \\
\bottomrule
\end{tabular}
\end{table}
As shown in Figure~\ref{fig:weak}, with reformulated event constraints, EXaCTz's runtime remains nearly constant up to 64 GPUs with over 80\% efficiency and achieves up to $33\times$ speedup over the original formulation. At 128 GPUs, EXaCTz reaches 55.6\% efficiency, compared with 6.4\% using the original formulation. The efficiency drop at 128 GPUs is primarily due to increased local correction iterations as the global problem size grows, while communication remains a small fraction of the total runtime.
The scalability improvement comes from eliminating explicit integral path computation, where EXaCTz enforces consistent min/max neighbors for extremum graphs (similar to pMSz~\cite{li2026pmsz}) and global critical point ordering for contour trees, avoiding path tracing.

Table~\ref{tab:multi_node} demonstrates EXaCTz's efficiency on large datasets. We process up to 512 GB data (Turbulence $4096^3$) using 128 GPUs, where the correction takes 15.66 seconds for SZ3, 18.22 seconds for cuSZp, and remains under 48 seconds even in the slowest case (47.56 seconds for cuZFP), while the OCR remains practical (e.g., 25.21 with SZ3).

\subsection{Extremum Graph and Contour Tree Preservation}

\begin{mdframed}
\noindent\textbf{Observation 5.}
Our approach achieves \textbf{100\%} preservation of extremum graphs and contour trees across all tested datasets.
\end{mdframed}
We evaluate the effectiveness of EXaCTz in preserving extremum graphs and contour trees under error-bounded lossy compression using scientific datasets that fit within a single GPU's memory. 

\begin{table}[htb]
\centering
\small
\caption{Topology preservation under a relative error bound of $10^{-4}$. SZ3 reports results before correction, while TopoA, pMSz, and EXaCTz report results after correction. All metrics are reported in terms of recall. N/A indicates results that failed to complete within a one-hour time limit.}
\vspace{-2mm}
\label{tab:metrics}
\setlength{\tabcolsep}{2pt} 
\begin{tabular}{l ccc ccc ccc ccc}
\toprule
\multirow{2.5}{*}{Dataset} & \multicolumn{3}{c}{Before Correction} & \multicolumn{9}{c}{After Correction} \\
\cmidrule(lr){2-4} \cmidrule(lr){5-13}
 & \multicolumn{3}{c}{SZ3} & \multicolumn{3}{c}{TopoA} & \multicolumn{3}{c}{pMSz} & \multicolumn{3}{c}{EXaCTz} \\
\cmidrule(lr){2-4} \cmidrule(lr){5-7} \cmidrule(lr){8-10} \cmidrule(lr){11-13}
 & CP & EG & CT & CP & EG & CT & CP & EG & CT & CP & EG & CT \\
\midrule
QMCPack       & 0.66 & 0.52 & 0.50 & \textbf{1.00} & 0.90 & \textbf{1.00} & 0.99 & 0.72 & 0.61 & \textbf{1.00} & \textbf{1.00} & \textbf{1.00} \\
AT            & 0.88 & 0.87 & 0.79 & \textbf{1.00} & 0.96 & \textbf{1.00} & 0.92 & 0.99 & 0.81 & \textbf{1.00} & \textbf{1.00} & \textbf{1.00} \\
Vortex        & 0.97 & 0.88 & 0.78 & \textbf{1.00} & 0.97 & \textbf{1.00} & 0.99 & 0.99 & 0.82 & \textbf{1.00} & \textbf{1.00} & \textbf{1.00} \\
Turbulence    & 0.90 & 0.87 & 0.76 & N/A  & N/A  & N/A  & 0.99 & 0.99 & 0.77 & \textbf{1.00} & \textbf{1.00} & \textbf{1.00} \\
NYX ($512^3$) & 0.86 & 0.77 & 0.52 & N/A  & N/A  & N/A  & 0.98 & 0.92 & 0.62 & \textbf{1.00} & \textbf{1.00} & \textbf{1.00} \\
Combustion    & 0.77 & 0.32 & 0.20 & \textbf{1.00} & 0.96 & \textbf{1.00} & 0.99 & 0.95 & \textbf{1.00} & \textbf{1.00} & \textbf{1.00} & \textbf{1.00} \\
\bottomrule
\end{tabular}
\end{table}

\begin{figure}[htb]
    \vspace{-2mm}
    \centering
    \includegraphics[width=\linewidth]
    {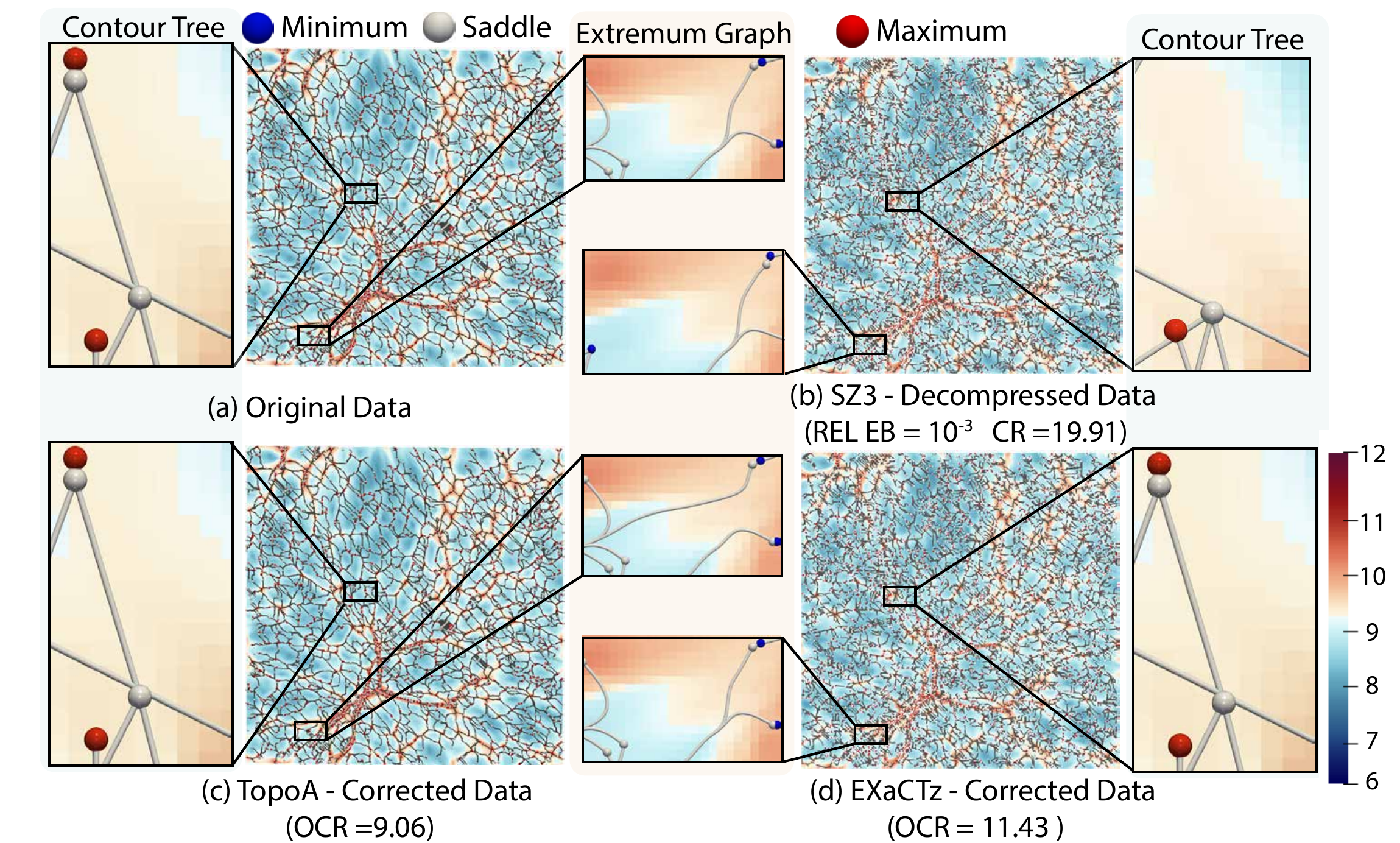}
    \vspace{-6mm}
    \caption{Qualitative comparison of topological structures under lossy compression. (a) Original data. (b) SZ3-decompressed data. SZ3-decompressed data corrected by (c) TopoA and (d) EXaCTz.}
    \label{fig:EG&CT}
    \vspace{-0.1in}
\end{figure}

\noindent\textbf{Quantitative Comparison with Baseline Methods.}
We use SZ3 as the base compressor and compare EXaCTz against two state-of-the-art topology-preserving methods: TopoA and pMSz~\cite{li2026pmsz} (a method that preserves Morse-Smale segmentation).  As shown in Table~\ref{tab:metrics}, EXaCTz consistently achieves 1.00 recall across all metrics and datasets, whereas existing methods only achieve partial preservation. While TopoA guarantees contour tree preservation (1.00 CT recall), it yields errors in the extremum graph (e.g., 0.96 EG recall on AT and Combustion). pMSz primarily preserves extrema and their connectivity but neglects saddles, leaving its CP, EG, and CT recall scores below 1.00. Without post-processing, standard SZ3 exhibits substantial topology loss across all datasets. \updated{EXaCTz achieves 1.00 precision across all metrics and datasets; full results are in Appendix~F (Table 3).}

\noindent\textbf{Qualitative Visual Comparison with SZ3 and TopoA.}
We extract the extremum graph and contour tree of a 2D NYX slice ($512^2$) as shown in Figure~\ref{fig:EG&CT}. Consistent with Table~\ref{tab:metrics}, EXaCTz corrects all SZ3-induced distortions (Figure~\ref{fig:EG&CT}(b)) across both the extremum graph and the contour tree (Figure~\ref{fig:EG&CT}(d)), whereas TopoA still leaves residual distortions in the extremum graph (Figure~\ref{fig:EG&CT}(c)).

\section{Conclusion and Future Work}

In this paper, we address three key limitations of existing topology-preserving compressors: the high-throughput performance gap, limited support for diverse topological descriptors, and the lack of explicit theoretical guarantees.
We propose EXaCTz, a distributed- and GPU-parallel algorithm that concurrently preserves extremum graphs and full contour trees under error-bounded lossy compression. By eliminating explicit topology construction and enforcing localized scalar ordering constraints, EXaCTz achieves up to $213\times$ speedup over TopoA~\cite{gorski2025general} in serial CPU settings and GB/s-scale throughput on GPUs. 
To support both extremum graphs and contour trees, we leverage their theoretical connection to enforce a set of constraints. Additionally, we formally analyze the iterative correction process and derive an explicit iteration upper bound, ensuring predictable convergence.
By reformulating event constraints for distributed-memory environments, we improve weak-scaling efficiency from 6.4\% to 55.6\% at 128 GPUs.

\noindent\textbf{Limitations.}
\updated{We identify a few limitations of EXaCTz. First, maintaining the necessary data fields and topological constraints introduces an approximate $8\times$ memory overhead, posing a practical challenge for memory-bound \textit{in-situ} applications. Second, the theoretical 
upper bound for maximum correction iterations is relatively loose. Third, while reformulated event constraints avoid expensive integral path tracing, global exchanges of critical point scalar values are 
still required to detect ordering violations. 
Fourth, our method supports simplification as an optional preprocessing step (Section~\ref{sec:serial_cpu}) but does not 
incorporate it into the full pipeline due to its computational cost and scalability limitations. Fifth, the monotonic decreasing edits introduce a slight negative bias in the error distribution (Appendix~G), which may also affect integral conservation quantities such as total mass, energy, and momentum. Reducing memory overhead, establishing tighter 
iteration bounds, optimizing communication strategies, integrating scalable simplification, and mitigating error bias are directions for future work.} 

\acknowledgments{This work was supported by the National Science Foundation (NSF) under Awards OAC-2313122, OAC-2313123, OAC-2313124, and OAC-2628473, and by the U.S. Department of Energy (DOE), Office of Science, Advanced Scientific Computing Research (ASCR), under Contract DE-SC0025677. This research also used resources of the National Energy Research Scientific Computing Center (NERSC), a DOE Office of Science User Facility.}

\bibliographystyle{abbrv-doi-hyperref-narrow}
\bibliography{refs-EXaCTz}

\clearpage
\appendix
This appendix provides supplementary material to support the main findings of this paper. Section A details the evaluation datasets. Section B presents additional performance evaluation using TTHRESH and SPERR as base compressors. Section C provides empirical analyses on parameter N. Section D analyzes the topological distortion introduced by different base compressors. Section E compares the overall compression ratio with TopoA. Section F reports the precision of topological features. Section G analyzes error distributions before and after correction. Section H compares edit maps with and without integral path computation. Section I conducts weak scaling using downscaled real datasets.
\section{DETAILS OF THE DATASETS}\label{datasets}
This section outlines the datasets used in our evaluation.

\noindent The \textbf{QMCPACK} dataset is from performance tests of a continuum quantum Monte Carlo simulation~\cite{QMCPack} (accessed via the SDRBench~\cite{SDRBench}, \url{https://sdrbench.github.io/}).

\noindent The \textbf{Adenine Thymine} is a molecular simulation dataset from the Topology ToolKit (TTK) Tutorial Data~\cite{ttk} (\url{https://topology-tool-kit.github.io/tutorials.html}), which models the electron density of a DNA base pair. It restricts the scalar field to a 2D planar slice within a 3D domain, reflecting common practices in quantum chemistry to reduce complexity while retaining essential structural information.

\noindent The \textbf{Turbulent Vortex} dataset is a standard benchmark widely used in topological data analysis and originates from a fluid turbulence simulation and features intricate, interlocking vortex structures (accessed via the SDRBench~\cite{SDRBench}, \url{https://sdrbench.github.io/}).

\noindent The \textbf{NYX} dataset~\cite{Nyx} consists of 3D spatial fields generated by the cosmological hydrodynamics code NYX~\cite{Almgren_2013}. Each output is a three-dimensional spatial field used for post-analysis and includes variables such as dark matter density and temperature. For our evaluation, we use this dataset at three different dimensions: $512^3$, $1024^3$, and $2048^3$ (accessed via the SDRBench~\cite{SDRBench}, \url{https://sdrbench.github.io/}).

\noindent The \textbf{Turbulent Combustion Simulation} dataset captures a turbulent non-premixed flame with detailed chemistry, generated via high-fidelity direct numerical simulation (DNS)~\cite{Chen2009}. For our evaluation, we evaluate this dataset in three dimensions: $256^3$, $560^3$, and $4096^3$.

\section{Additional Performance Evaluation on TTHRESH and SPERR}

We extend our performance evaluation by using TTHRESH~\cite{TTHRESH} and SPERR~\cite{SPERR} as base compressors across single- and multi-GPU environments, as shown in Table~\ref{tab:tthresh}. In the multi-GPU setting, the base compressors operate independently on each local data block. Note that TTHRESH does not support strict pointwise error bounding; we tune its parameters so that the resulting maximum pointwise error is close to the target bound of $10^{-4}$.

Across all datasets, EXaCTz consistently completes correction within the same order of magnitude as the base compression time, and often substantially faster, and the edit ratio remains consistently low (typically below 8\%). For example, on the $2048^3$ NYX dataset, EXaCTz requires 18.43s for correction compared to 437s for the initial SPERR compression, while modifying only 2.86\% of the data. For smaller datasets, the correction time is negligible (e.g., 0.03–0.27s for QMCPack and AT), while for larger datasets, the correction cost scales moderately with data size.

Comparing compressors, correcting TTHRESH generally requires more iterations and longer runtime than SPERR. This is because TTHRESH’s global tensor decomposition introduces more spatially distributed perturbations, leading to widespread violations that require additional correction iterations.

\begin{table}
\centering
\small
\caption{Performance evaluation of EXaCTz on NVIDIA A100 GPUs using SPERR and TTHRESH as base compressors (relative error bound $\xi=10^{-4}$). SPERR and TTHRESH are CPU-based compressors, and the reported EXaCTz time refers only to the correction stage. \#GPUs denotes the minimum number of GPUs required to fit each dataset.}
\renewcommand{\arraystretch}{0.9}
\label{tab:tthresh}
\setlength{\tabcolsep}{0.8pt}
\begin{tabular}{lccccc|cccc}
\toprule
\multirow{2}{*}{Dataset} 
& \multirow{2}{*}{\shortstack{Size\\(GB)}} 
& \multirow{2}{*}{\#GPUs} 
& \multicolumn{3}{c|}{\shortstack{Stage 1\\Compression-only}} 
& \multicolumn{4}{c}{\shortstack{Stage 2\\EXaCTz Correction-only}} \\
\cmidrule(lr){4-6} \cmidrule(lr){7-10}
& & & Type & Time (s) & CR & Time (s) & Edit \% & Iter. & OCR \\
\midrule

\multirow{2}{*}{QMCPack} & \multirow{2}{*}{0.004} & \multirow{2}{*}{1}
& SPERR   & 0.30   & 89.16  & 0.05  & 0.84 & 51  & 51.06 \\
& & & TTHRESH & 0.22   & 88.04  & 0.03  & 0.77 & 34  & 71.11 \\
\midrule

\multirow{2}{*}{AT} & \multirow{2}{*}{0.006} & \multirow{2}{*}{1}
& SPERR   & 0.54   & 73.78  & 0.01  & 0.27 & 13  & 69.21 \\
& & & TTHRESH & 0.27   & 153.61 & 0.04  & 0.01 & 56  & 65.47 \\
\midrule

\multirow{2}{*}{Vortex} & \multirow{2}{*}{0.015} & \multirow{2}{*}{1}
& SPERR   & 1.12   & 20.41  & 1.30  & 1.30 & 52  & 17.84 \\
& & & TTHRESH & 0.94   & 7.36   & 0.06  & 0.99 & 133 & 7.02 \\
\midrule

\multirow{2}{*}{Turbulence} & \multirow{2}{*}{0.125} & \multirow{2}{*}{1}
& SPERR   & 10.21  & 9.22   & 0.18  & 4.87 & 36  & 6.99 \\
& & & TTHRESH & 6.53   & 4.85   & 0.28  & 4.22 & 159 & 4.19 \\
\midrule

\multirow{2}{*}{NYX ($512^3$)} & \multirow{2}{*}{1.000} & \multirow{2}{*}{1}
& SPERR   & 11.34  & 11.92  & 0.53  & 2.75 & 56  & 10.20 \\
& & & TTHRESH & 47.02  & 6.98   & 0.71  & 1.63 & 117 & 6.41 \\
\midrule

\multirow{2}{*}{Combustion} & \multirow{2}{*}{1.310} & \multirow{2}{*}{1}
& SPERR   & 71.19  & 145.11 & 0.45  & 0.87 & 79  & 104.63 \\
& & & TTHRESH & 46.87  & 64.19  & 1.36  & 0.24 & 287 & 60.37 \\
\midrule

\multirow{2}{*}{NYX ($1024^3$)} & \multirow{2}{*}{8.00} & \multirow{2}{*}{2}
& SPERR   & 441    & 12.01  & 5.46  & 2.65 & 152 & 10.42 \\
& & & TTHRESH & 137.64 & 8.11   & 8.74  & 2.89 & 453 & 7.34 \\
\midrule

\multirow{2}{*}{NYX ($2048^3$)} & \multirow{2}{*}{64.00} & \multirow{2}{*}{16}
& SPERR   & 437    & 12.45  & 18.43 & 2.86 & 185 & 10.81 \\
& & & TTHRESH & 137.35 & 12.09  & 32.36 & 7.71 & 372 & 9.79 \\
\midrule

\multirow{2}{*}{Turbulence ($4096^3$)} & \multirow{2}{*}{512.00} & \multirow{2}{*}{128}
& SPERR   & 447    & 48.44  & 15.66 & 1.76 & 447 & 34.74 \\
& & & TTHRESH & 133.26     & 26.58     & 46.72   & 1.69   & 568  & 16.44 \\
\bottomrule
\end{tabular}
\end{table}

\section{Empirical Evaluation on Parameter N}

\begin{figure}
    \centering
    \includegraphics[width=\linewidth]{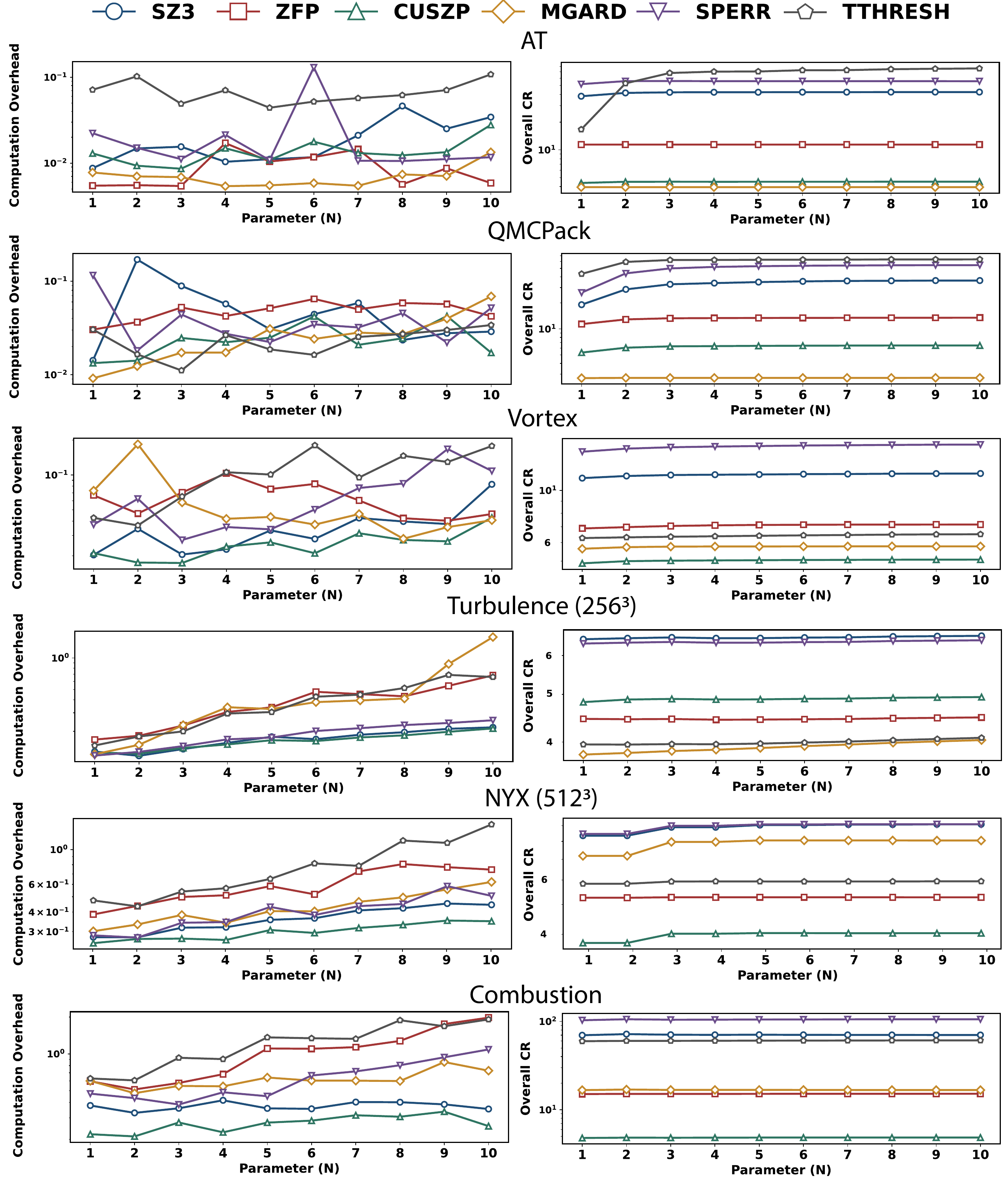}
    \caption{Trade-off between compression ratio (CR) and computation time with respect to parameter $N$ across different datasets and compressors.}
    \label{fig:Parameter}
\end{figure}

We evaluate the impact of the parameter $N$, which limits the maximum number of edits allowed per vertex, on both computational overhead and overall compression ratio (OCR), as shown in Figure~\ref{fig:Parameter}.

Across all datasets and compressors, $N$ introduces a trade-off between correction cost and overall compression ratio. When $N$ is small, the algorithm restricts the number of edits per vertex, resulting in lower computational overhead but reduced CR. Conversely, increasing $N$ allows for a higher number of finer edits to be applied to each vertex within each iteration. While this provides improved OCR, it also leads to a higher overall correction time.

A consistent trend emerges despite variations across datasets and compressors: the improvement in CR gradually saturates as $N$ increases. In most cases, OCR shows only marginal gains when $N > 5$, whereas the computational overhead continues to increase. 
While the optimal value of $N$ may vary slightly across datasets and compressors, $N=5$ performs well across all evaluated cases. Based on this observation, we set $N=5$ in our experiments as a practical balance between correction cost and compression ratio.

\begin{table*}
\centering
\small
\caption{Topology preservation comparison under relative error bound $10^{-4}$. ZFP, TTHRESH, SPERR, MGARD, and cuSZp represent the baselines before correction. All metrics are reported in terms of recall.}
\vspace{-3mm}
\label{tab:metrics1}
\renewcommand{\arraystretch}{1.4}
 
\begin{tabular}{l ccc ccc ccc ccc ccc @{\quad} ccc}
\toprule
\multirow{2.5}{*}{Dataset} & \multicolumn{15}{c}{Before Correction} & \multicolumn{3}{c}{After Correction} \\
\cmidrule(lr){2-16} \cmidrule(lr){17-19}

 & \multicolumn{3}{c}{ZFP} & \multicolumn{3}{c}{TTHRESH} & \multicolumn{3}{c}{SPERR} & \multicolumn{3}{c}{cuSZp} & \multicolumn{3}{c}{MGARD} & \multicolumn{3}{c}{EXaCTz} \\
\cmidrule(lr){2-4} \cmidrule(lr){5-7} \cmidrule(lr){8-10} \cmidrule(lr){11-13} \cmidrule(lr){14-16} \cmidrule(lr){17-19}
 & CP & EG & CT
 & CP & EG & CT 
 & CP & EG & CT 
 & CP & EG & CT 
& CP & EG & CT 
& CP & EG & CT\\
\cmidrule(lr){2-4} \cmidrule(lr){5-7} \cmidrule(lr){8-10} \cmidrule(lr){11-13} \cmidrule(lr){14-16} \cmidrule(lr){17-19}
QMCPack     & 0.93 & 0.84 & 0.84
            & 0.91 & 0.86 & 0.76 
            & 0.80 & 0.72 & 0.64 
            & 0.27 & 0.48 & 0.40 
            & 0.91 &0.80 & 0.82 
            & \textbf{1.00} & \textbf{1.00} & \textbf{1.00}
            \\
AT          & 0.98 & 0.94 & 0.91
            & 0.75& 0.68& 0.61 
            & 0.89 & 0.88& 0.84 
            &  0.68 & 0.64 & 0.50 
            & 0.90& 0.94& 0.92 
            & \textbf{1.00} & \textbf{1.00} & \textbf{1.00}
            \\
Vortex      & 0.99 & 0.99 & 0.95
            & 0.99 & 0.88 & 0.93 
            & 0.98& 0.94& 0.82 
            & 0.95 & 0.89 & 0.74 
            & 0.99 & 0.97 & 0.89 
            & \textbf{1.00} & \textbf{1.00} & \textbf{1.00}
            \\
Turbulence  & 0.91 & 0.91 & 0.83
            & 0.99 & 0.99 & 0.83 
            & 0.98 & 0.97 &0.77 
            & 0.97 & 0.95 & 0.73 
            & 0.99 & 0.99 & 0.82 
            & \textbf{1.00} & \textbf{1.00} & \textbf{1.00}
            \\
NYX         & 0.99 & 0.98 & 0.77 
            & 0.97 & 0.96 & 0.72
            & 0.90& 0.82& 0.57
            & 0.80 & 0.69 & 0.44 
            & 0.85 & 0.75 & 0.51 
            & \textbf{1.00} & \textbf{1.00} & \textbf{1.00}
            \\
Combustion  & 0.77 & 0.32 & 0.20
            & 0.96&  0.84& 0.61 
            & 0.87& 0.54 & 0.35
            & 0.52 & 0.27 & 0.25 
            & 0.82 & 0.37 & 0.29 
            & \textbf{1.00} & \textbf{1.00} & \textbf{1.00}
            \\

\hline
\end{tabular}
\vspace{-2mm}
\end{table*}

\begin{table*}[htb!]
\centering
\caption{\updated{Topology preservation under a relative error bound of $10^{-4}$. SZ3 reports results before correction, while TopoA, pMSz, and EXaCTz report results after correction. All metrics are reported in terms of precision. MS denotes Morse--Smale segmentation-label precision. N/A indicates results that were not evaluated or failed to complete within a one-hour time limit.}}
\vspace{-3mm}
\label{tab:metrics_precision}
\begin{tabular}{l cccc cccc cccc cccc}
\toprule
\multirow{2.5}{*}{Dataset} & \multicolumn{4}{c}{Before Correction} & \multicolumn{12}{c}{After Correction} \\
\cmidrule(lr){2-5} \cmidrule(lr){6-17}
 & \multicolumn{4}{c}{SZ3} & \multicolumn{4}{c}{TopoA} & \multicolumn{4}{c}{pMSz} & \multicolumn{4}{c}{EXaCTz} \\
\cmidrule(lr){2-5} \cmidrule(lr){6-9} \cmidrule(lr){10-13} \cmidrule(lr){14-17}
 & CP & EG & CT & MS & CP & EG & CT & MS & CP & EG & CT & MS & CP & EG & CT & MS \\
\midrule
QMCPack       & 0.66 & 0.52 & 0.50 & 0.87 & 0.98 & 0.90 & \textbf{1.00} & 0.98 & \textbf{1.00} & 0.69 & 0.42 & \textbf{1.00} & \textbf{1.00} & \textbf{1.00} & \textbf{1.00} & \textbf{1.00} \\
AT            & 0.83 & 0.85 & 0.74 & 0.95 & \textbf{1.00} & 0.96 & \textbf{1.00} & \textbf{1.00} & \textbf{1.00} & 0.93 & 0.22 & \textbf{1.00} & \textbf{1.00} & \textbf{1.00} & \textbf{1.00} & \textbf{1.00} \\
Vortex        & 0.96 & 0.88 & 0.78 & 0.90 & \textbf{1.00} & \textbf{1.00} & \textbf{1.00} & \textbf{1.00} & \textbf{1.00} & \textbf{1.00} & 0.83 & \textbf{1.00} & \textbf{1.00} & \textbf{1.00} & \textbf{1.00} & \textbf{1.00} \\
Turbulence    & 0.98 & 0.96 & 0.76 & 0.91 & N/A  & N/A  & N/A  & N/A & \textbf{1.00} & \textbf{1.00} & 0.77 & \textbf{1.00} & \textbf{1.00} & \textbf{1.00} & \textbf{1.00} & \textbf{1.00} \\
NYX ($512^3$) & 0.84 & 0.76 & 0.51 & 0.75 & N/A  & N/A  & N/A  & N/A  & \textbf{1.00} & 0.99 & 0.63 & \textbf{1.00} & \textbf{1.00} & \textbf{1.00} & \textbf{1.00} & \textbf{1.00} \\
Combustion    & 0.76 & 0.22 & 0.14 & 0.39 & 0.99 & 0.98 & \textbf{1.00} & 0.96 & \textbf{1.00} & 0.95 & \textbf{1.00} & \textbf{1.00} & \textbf{1.00} & \textbf{1.00} & \textbf{1.00} & \textbf{1.00} \\
\bottomrule
\end{tabular}
\vspace{-4mm}
\end{table*}

\section{Topological Distortion by Base Compressors}
To evaluate how different base compressors affect extremum graphs and contour trees, we tested them under the same relative error bound of $10^{-4}$, as shown in Table~\ref{tab:metrics1}. The evaluated compressors include ZFP (transform-based), TTHRESH (tensor-based), SPERR (block-based), MGARD, and cuSZp (cuZFP is excluded because it does not support fixed error bounds). In contrast, after the iterative correction by EXaCTz, both EG and CT are perfectly preserved across all datasets and base compressors, achieving a recall of 1.00.

Across all compressors, CP recall remains relatively high compared to EG and CT, indicating that critical point detection is more robust under error-bounded compression. In contrast, EG and CT exhibit significantly larger distortion, especially on complex datasets such as Combustion and NYX. For example, on the Combustion dataset, CT recall drops to as low as 0.25 for cuSZp and 0.29 for MGARD. Similar trends are observed on NYX, where EG and CT recalls decrease sharply for most compressors.

Different compressors also exhibit different distortion patterns. ZFP and MGARD tend to better preserve EG structures in some cases (e.g., ZFP achieves 0.99 EG recall on the Vortex dataset), but still distort CT due to global structural changes. SZ3/cuSZp often preserve CP reasonably well but can severely distort EG and CT, particularly on datasets with complex structures. Notably, cuSZp shows strong degradation in CT (e.g., 0.25 on Combustion), indicating that high-throughput GPU compressors may introduce nontrivial topological inconsistencies. SPERR and TTHRESH exhibit intermediate behavior, with moderate degradation across EG and CT.

\section{Additional Comparison with TopoA}
\begin{figure}[htb!]
    \centering
    \includegraphics[width=\linewidth]{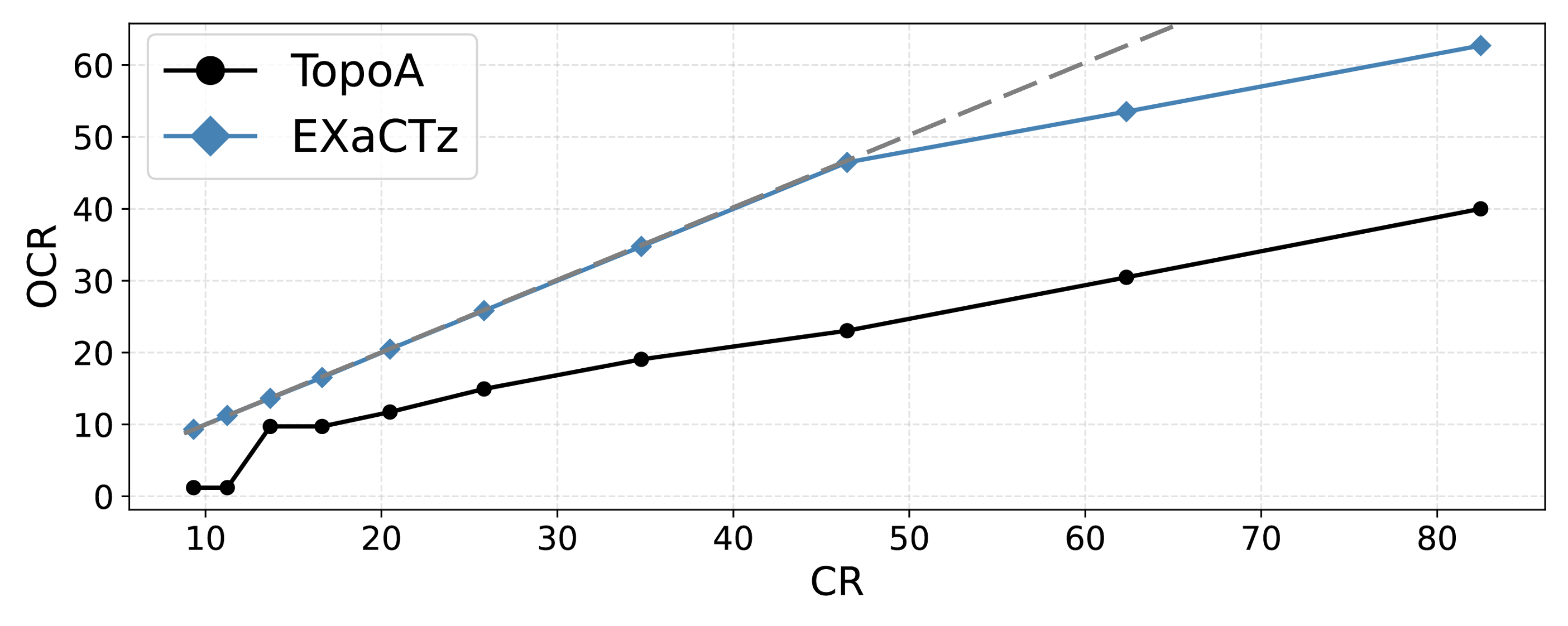}
    \caption{\updated{OCR vs.\ CR for EXaCTz and TopoA on the AT dataset using 
SZ3 as the base compressor across varying error bounds. The dashed line indicates OCR = CR.}}
    \label{fig:ocr_cr}
\end{figure}
\updated{Figure~\ref{fig:ocr_cr} compares the overall compression ratio (OCR) against the base compression ratio (CR) for EXaCTz and TopoA on the AT dataset using SZ3 as the base compressor across varying error bounds. The gap between CR and OCR reflects the storage overhead introduced by topology correction. As the error bound increases (higher CR), both methods require more storage overhead to preserve topology. EXaCTz consistently achieves higher OCR than TopoA across all compression ratios, indicating lower correction overhead. The dashed line (OCR = CR) represents the ideal case with zero correction overhead. EXaCTz closely tracks this line up to CR $\approx$ 45, after which the gap widens as larger error bounds introduce more topological distortions requiring additional corrections.}

\section{Precision of Topological Features}
\updated{Table~\ref{tab:metrics_precision} reports the precision of critical points (CP), extremum graphs (EG), contour trees (CT), and Morse--Smale segmentations (MS) under a relative error bound of $10^{-4}$. Precision is defined as the fraction of topological features in the 
corrected data that are also present in the original data.}

\updated{EXaCTz achieves 1.00 precision across all metrics and datasets. Note that EXaCTz also achieves 1.00 MS precision because our EG constraints preserve the steepest ascending and descending neighbors of every vertex, which preserves all integral lines and thus the induced MS segmentation. TopoA achieves 1.0 CT precision but lower EG precision (e.g., 0.90 on QMCPack and 0.96 on AT and Combustion). pMSz shows lower precision across most 
metrics, particularly for EG (e.g., 0.42 on QMCPack) and CT 
(e.g., 0.22 on QMCPack), as it targets Morse--Smale segmentation preservation and does not explicitly enforce extremum graph or contour tree consistency. Before correction, SZ3 exhibits low precision across all metrics, confirming that lossy compression introduces distorted topological features.}

\section{Error Distribution Analysis}
\begin{figure}[htb!]
    \centering
    \includegraphics[width=\linewidth]{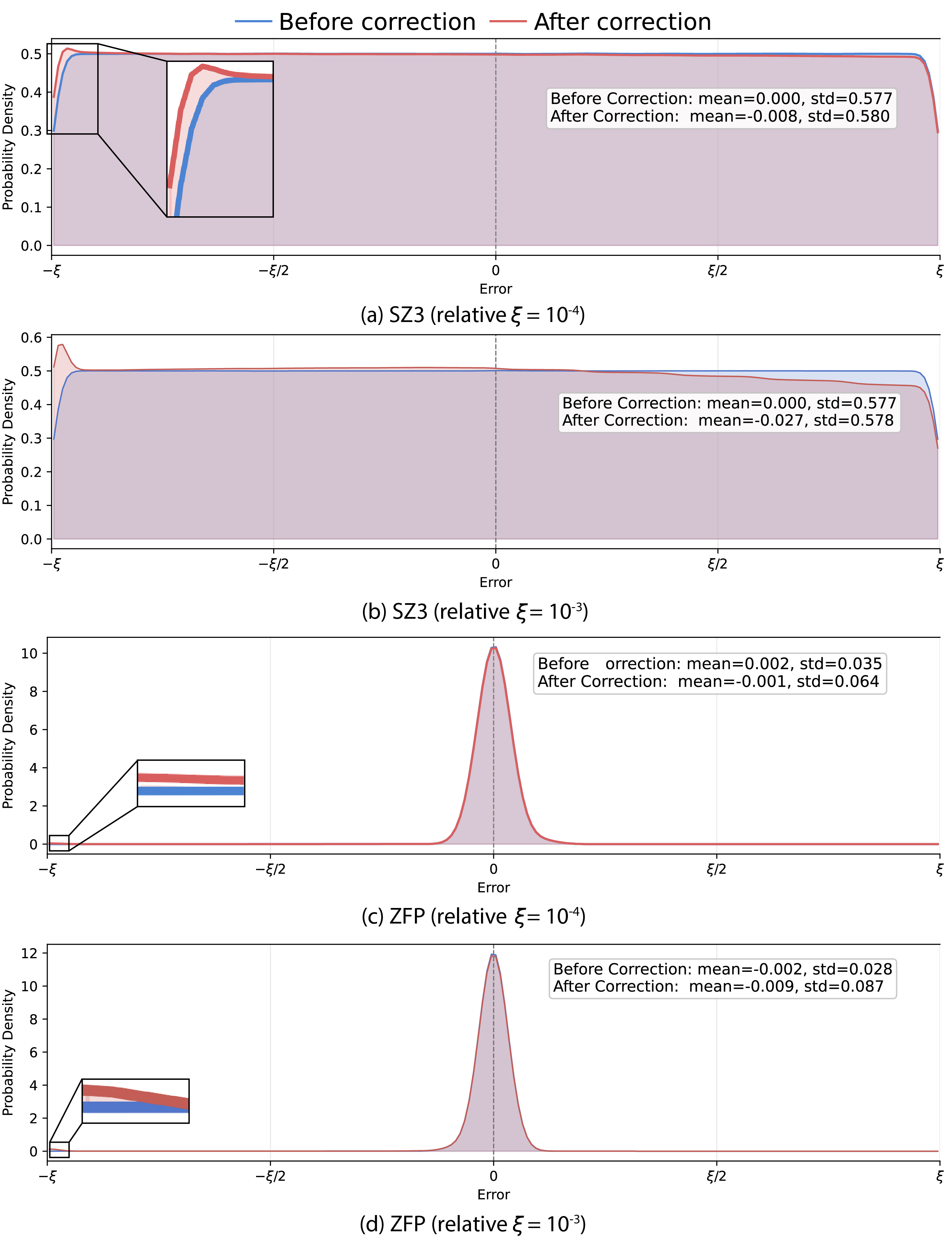}
    \caption{Error distributions before and after EXaCTz correction on the NYX ($512^3$) dataset with (a, b) SZ3 and (c, d) ZFP as base compressors (relative error bound $\xi = 10^{-4}$ and $\xi = 10^{-3}$). Insets show the zoomed-in view near $-\xi$.}
    \label{fig:error_dist}
\end{figure}

\updated{We analyze the error distributions before and after EXaCTz correction to quantify the effect of monotonic decreasing edits on the error distribution. Figure~\ref{fig:error_dist} shows the results on the NYX ($512^3$) dataset using SZ3 and ZFP as base compressors under two relative error bounds ($\xi = 10^{-4}$ and $\xi = 10^{-3}$).}

\updated{The before-correction error distributions are consistent with known 
compressor characteristics~\cite{Lindstrom2017ErrorDist}: near-uniform for SZ3 and approximately normal for ZFP. After correction, both compressors show a slight negative shift in mean (e.g., from $0.000$ to $-0.008$ for SZ3 and from $0.002$ to $-0.001$ for ZFP at $\xi = 10^{-4}$). The bias increases with larger error bounds (e.g., SZ3 mean shifts to $-0.027$ at $\xi = 10^{-3}$), as more topological distortions require more edits. For SZ3, as shown in Figure~\ref{fig:error_dist} (a) and (b), the after-correction error distribution accumulates near $-\xi$, indicating that more vertices are driven to their lower bound $f_i - \xi$.}

\section{Edit Map Comparison}

\begin{figure}[htb!]
    \centering
    \vspace{-2mm}
    \includegraphics[width=0.8\linewidth]{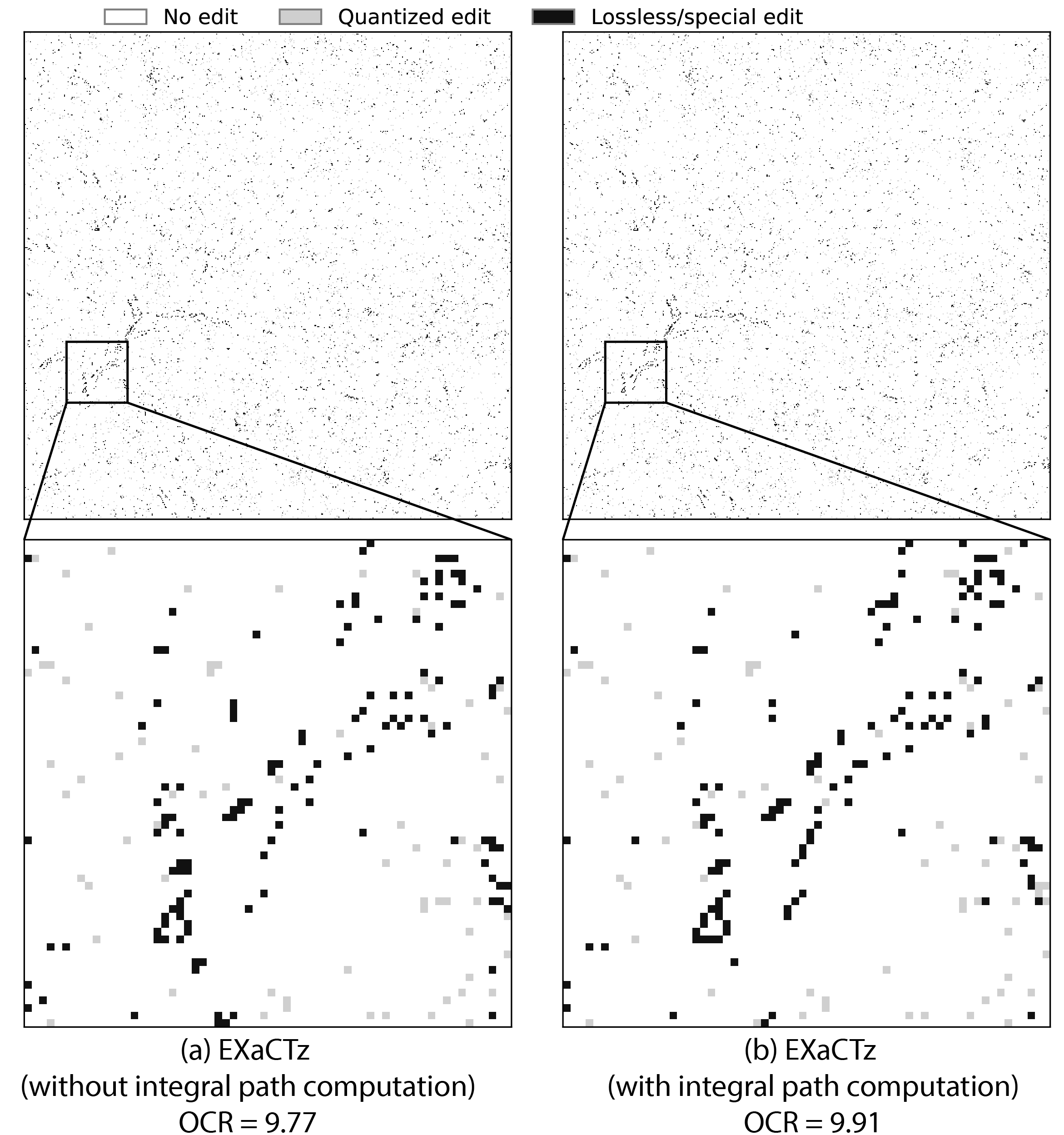}
    \vspace{-2mm}
    \caption{Edit maps of EXaCTz (a) without and (b) with integral path computation on a 2D NYX slice (SZ3, $\xi = 10^{-4}$). Insets show the zoomed-in view.}
    \label{fig:edit_map}
\end{figure}
\updated{We compare the edit maps of EXaCTz with and without integral path computation to visualize the impact on the edit pattern and storage overhead. Figure~18 visualizes the edit maps on a 2D NYX slice (SZ3, $\xi = 10^{-4}$). While some edits appear in only one of the two maps, this is because the two formulations enforce different constraint sets: the original event constraints target violations along integral paths, whereas the reformulated constraints enforce global critical point ordering. As a result, the spatial distribution of edits differs, but the total number of edits is nearly identical (4,902,506 vs.\ 4,904,024, a difference of 1,518), with the reformulated version producing slightly more edits overall, consistent with its over-approximation of the original constraints. Avoiding integral path computation introduces a slight OCR overhead (9.77 vs.\ 9.91) but reduces the correction time from 0.492s to 0.254s ($1.94\times$).}

\section{Weak Scaling on Downscaled Real Data}
\begin{figure}[htb]
    \centering
    \includegraphics[width=\linewidth]{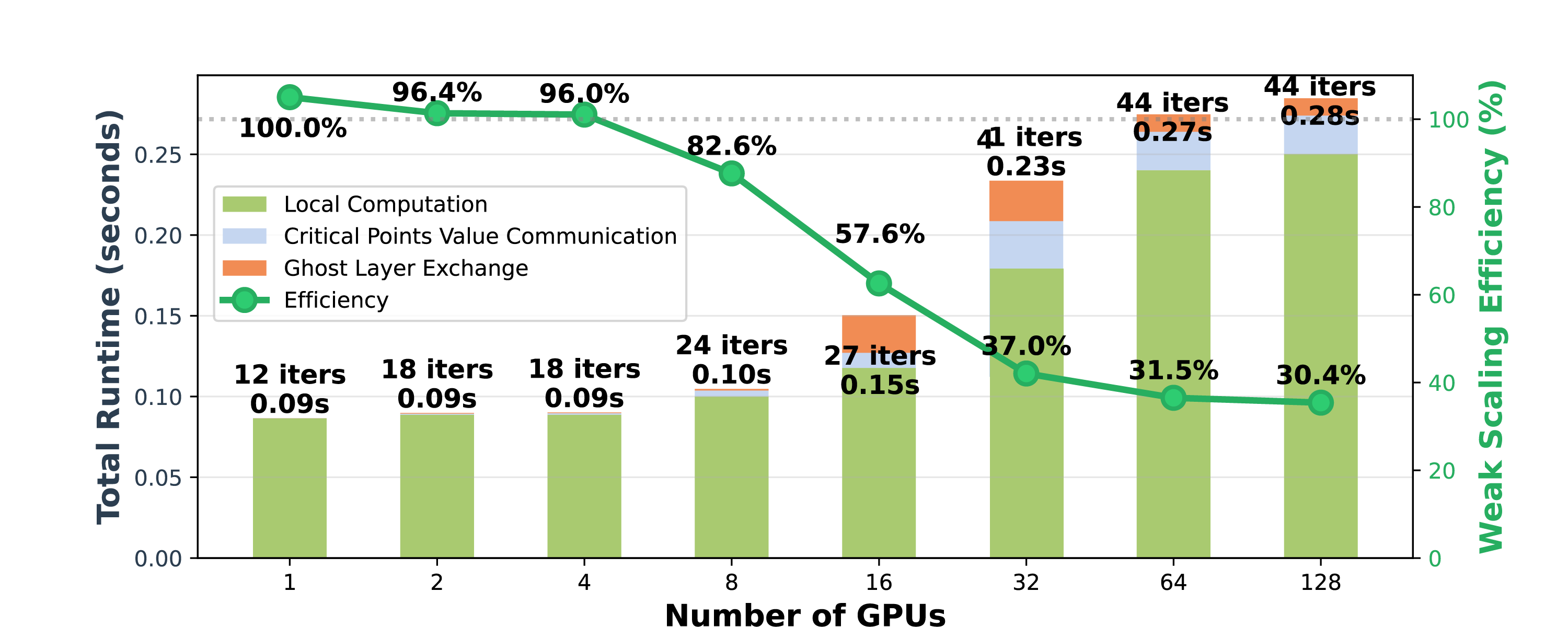}
    \caption{Weak scaling on downscaled real Turbulence data from 1 to 
    128 GPUs (SZ3, $\xi = 10^{-4}$). Stacked bars show the time 
    breakdown; the line shows weak scaling efficiency.}
    \label{fig:weak_real}
\end{figure}
To complement the upsampled weak scaling evaluation, we conducted weak scaling using downscaled versions of the real Turbulence ($4096^3$) dataset. As shown in Figure~\ref{fig:weak_real}, the iteration count varies across scales (12 at 1 GPU vs.\ 44 at 128 GPUs), as larger global problem sizes introduce more topological distortions. The increased local computation from additional iterations dominates the runtime, reducing weak scaling efficiency from 100\% at 1 GPU to 30.4\% at 128 GPUs.

\end{document}